\documentclass[lettersize,journal]{IEEEtran}

\IEEEoverridecommandlockouts

\usepackage{longtable}

\usepackage{float} 
\usepackage{adjustbox}
\usepackage{cite}
\usepackage{amsmath,amssymb,amsfonts}
\usepackage{graphicx}
\usepackage{textcomp}
\usepackage{xcolor}
\usepackage{algorithm} 
\usepackage{algpseudocode}
\usepackage{amsthm}

\newtheorem{proposition}{Proposition}
\newtheorem{lemma}{Lemma}
\newtheorem{definition}{Definition}
\usepackage{algorithm}
\usepackage{algcompatible}
\algnewcommand\algorithmicreturn{\mathbf{return}}
\algnewcommand\RETURN{\algorithmicreturn}

\usepackage{bm}

\usepackage{array}
\usepackage{stfloats}

\def\BibTeX{{\rm B\kern-.05em{\sc i\kern-.025em b}\kern-.08em
    T\kern-.1667em\lower.7ex\hbox{E}\kern-.125emX}}

\begin{document}
\bstctlcite{IEEEexample:BSTcontrol}

\title{Minimizing Energy Consumption in MU-MIMO via Antenna Muting by Neural Networks with Asymmetric Loss \\
\thanks{This work has been submitted to the IEEE for possible publication. Copyright may be transferred without notice, after which this version may no longer be accessible.}
}

 \author{
 \IEEEauthorblockN{Nuwanthika~Rajapaksha$^\dagger$, Jafar Mohammadi$^*$, Stefan Wesemann$^*$, Thorsten Wild$^*$, Nandana Rajatheva$^\dagger$}
  \small \\
$^\dagger$ University of Oulu, Finland\\
$^*$ Nokia Bell Labs, Nokia, Germany 
 }

\maketitle

\begin{abstract}
Transmit antenna muting (TAM) in multiple-user multiple-input multiple-output (MU-MIMO) networks allows reducing the power consumption of the base station (BS) by properly utilizing only a subset of antennas in the BS. In this paper, we consider the downlink transmission of an MU-MIMO network where TAM is formulated to minimize the number of active antennas in the BS while guaranteeing the per-user throughput requirements. To address the computational complexity of the combinatorial optimization problem, we propose an algorithm called neural antenna muting (NAM) with an asymmetric custom loss function. NAM is a classification neural network trained in a supervised manner. The classification error in this scheme leads to either sub-optimal energy consumption or lower quality of service (QoS) for the communication link. We control the classification error probability distribution by designing an asymmetric loss function such that the erroneous classification outputs are more likely to result in fulfilling the QoS requirements. 
Furthermore, we present three heuristic algorithms and compare them with the NAM. Using a 3GPP compliant system-level simulator, we show that NAM achieves $ \sim73\%$ energy saving compared to the full antenna configuration in the BS with $\sim95\%$ reliability in achieving the user throughput requirements while being around $1000\times$ and $24\times$ less computationally intensive than the greedy heuristic algorithm and the fixed column antenna muting algorithm, respectively.

\end{abstract}


\section{Introduction}
Considering the design criteria of 6G, \textit{sustainability} plays a more prominent role than in any other previous generations. Although the estimated CO2 equivalent emission for mobile networks is less than 0.5 percent of the global emissions \cite{Energy2011}, to achieve net-zero emission by 2050, the radio transmission optimization problem must be revisited. In this paper, we envision a slightly different problem statement for the well-known performance versus energy consumption trade-off, which allows for directly targeting net energy consumption in contrast to energy efficiency. 

Massive multiple-input multiple-output (MIMO) enables achieving higher rates in multi-user scenarios owing to the large number of transmit antennas in the base stations (BSs) \cite{6736761_mimo}. Using precoding and beamforming techniques, MIMO systems can serve multiple users simultaneously. 
Although the advantages of the massive MIMO systems are rooted in the large number of antennas \cite{Jacob2011}, this large number of antennas, their corresponding power amplifiers, and radio frequency (RF) chains lead to unnecessary power consumption. Especially when we assume a highly dynamic channel and load, it has been shown that to satisfy the users' quality of service (QoS) requirements only a fraction of the antennas are required \cite{8644331_nokia}. 
Transmit antenna muting (TAM) is a technique that reduces the BS power consumption by carefully deactivating a subset of antennas such that the desired system performance is satisfied.

The work in \cite{8644331_nokia, 6824974_xiang, 7569725_gao} investigates the achievable gains from TAM and different approaches for subarray antenna selection in MIMO. 
The authors in \cite{8644331_nokia} investigate practical deployment situations with varying user loads in massive MIMO, where they study the impact of different hybrid array architectures on spectral efficiency, average user throughput, and power consumption. 
Subarray TAM for massive MIMO is proposed in \cite{7569725_gao}, where the whole antenna array is segmented into a set of smaller subarrays and only one antenna is selected in each subarray, to reduce the hardware complexity of the RF switching network. 
The studies \cite{7127500_Benmimoune, 9172108_Akhtar, 7959971_dong} consider the joint user scheduling and antenna selection problem where only a subset of users for which the QoS constraints can be met are scheduled to share the same time-frequency resources. A greedy algorithm that removes the worst antenna elements and schedules users based on orthogonality is proposed in \cite{7127500_Benmimoune}, whereas a greedy two-step algorithm based on the largest minimum singular value is presented in \cite{7959971_dong}. TAM for improving energy 
efficiency is studied in \cite{6290313_jiang, Arash_2017}. In \cite{6290313_jiang}, the authors jointly optimize the transmit power, the number of active antennas, and the 
antenna subsets at the transmitter and receiver to maximize the energy efficiency by using exhaustive search along with sub-optimal algorithms to reduce the complexity. 
From a problem formulation perspective, \cite{6290313_jiang, Arash_2017} are closely akin to our objective in this work.

Recently, machine learning (ML)-based TAM is introduced to achieve a trade-off between the computational complexity and the system 
performance \cite{7523998_jingon, 8924932_elbir, 8446042_ibrahim, 9337188_Thang, 9448468_yu}. The authors in \cite{7523998_jingon} propose TAM as a multi-class classification approach using support vector machines and k-nearest neighbors 
method for a single-user MIMO (SU-MIMO) system. Furthermore, joint TAM and hybrid beamforming design for single-user mmWave MIMO is proposed in \cite{8924932_elbir} which uses two serial convolutional neural networks (CNNs). There, one CNN is used to predict the selected antennas, and the other CNN is used to 
estimate the hybrid beamformers. A deep neural network (DNN)-based approach for joint multicast beamforming and TAM is presented in \cite{8446042_ibrahim}, which learns the mapping between the channel state information and a subset of antennas that maximizes the minimum signal-to-noise ratio (SNR) at the receivers in a multicast setup. Furthermore, the authors in \cite{9337188_Thang} propose a joint antenna selection and precoding design (JASPD) algorithm to maximize the system sum rate of the MU-MIMO downlink subject to a transmit power constraint and QoS requirements. A DNN 
is implemented to learn the output from the JASPD algorithm to overcome its combinatorial complexity. Although \cite{9337188_Thang} applies ML to a TAM problem, our optimization problem, objective, and challenges are quite different. 
In \cite{9448468_yu}, the authors propose a low-complexity TAM algorithm for MU-MIMO based on multi-label learning, where a DNN is employed to determine the set of selected antennas for a given channel matrix. 
From a solution perspective, the work in \cite{9337188_Thang, 9448468_yu} can be considered the closest to ours as they both propose supervised learning-based solutions for TAM in MU-MIMO. However, our problem formulation is completely different from their sum rate optimization along with a fixed number of active antennas (RF chains).

In this paper, we investigate some practical solutions to reduce energy consumption by adaptively muting a subset of antennas in an MU-MIMO downlink setup. Our main constraint is to satisfy the per-user throughput requirements. We propose to represent our constrained optimization problem in terms of a supervised learning neural network (NN)-based method. We train an NN to learn the optimal antenna configuration given by a baseline algorithm. We further improve the reliability (throughput guarantee) of the NN output by introducing a novel asymmetric loss function for training the NN. We also propose a greedy-based TAM algorithm that selects the best antenna element subset in order to achieve the QoS requirements along with a few heuristics for TAM for performance comparison. Finally, we perform a detailed computational complexity analysis to show that the NN has significantly lower computational complexity compared to the greedy-based TAM algorithm and the other heuristics. Simulation results are presented showing TAM achieving energy savings, and specifically the potential of the ML-based approach 
for a practical implementation due to its reduced complexity and comparable performance.

Our main contributions: 
\begin{itemize}
    \item We target the total energy consumption minimization problem via TAM to minimize the number of active antenna elements, as opposed to energy efficiency maximization. Our problem statement is more aligned with 6G's vision of sustainability. We further consider per-user throughput requirement constraints in the optimization problem to achieve the desired network QoS performance.
    We introduce neural antenna muting (NAM) to address the computational complexity issue of the energy minimization problem. NAM is an NN paired with a designed \textit{asymmetric} loss function to achieve higher \textit{reliability}. 
    
    \item We mathematically prove the effectiveness of our proposed asymmetric loss on the \textit{reliability}. 
    \item We propose a greedy-based TAM algorithm and two low-complexity heuristics for TAM for performance comparison.
    \item We derive the complexity of the proposed NN and compare it with that of heuristic and greedy solutions.
\end{itemize}

The following notation is used throughout the paper. Boldface upper-case letters denote matrices and boldface lower-case letters denote column vectors. The superscripts $(\cdot)^T, (\cdot)^H$ and $(\cdot)^{-1}$
denote the transpose, Hermitian, and inverse operations, respectively. The $\lVert \cdot \rVert$ operator denotes the Euclidean norm of a vector and $\lVert \cdot \rVert_0$ denotes the $l_0$ norm of a vector. The $\lvert \cdot \rvert$ operator denotes the absolute value. Furthermore, $\mathbb{E} \big[\cdot \big]$ denotes the expectation of the given variable. By $diag(\mathbf{x})$ we denote a diagonal matrix generated with the diagonal given by the vector $\mathbf{x}$. An $N \times N$ identity matrix is denoted by $\mathbf{I}_N$. Finally, $tr( \mathbf{X} )$ denotes the trace of matrix $\mathbf{X}$.

\section{MU-MIMO System Model}
\label{section_system_model}

In this section, we present the general MU-MIMO system model without considering TAM for simplicity. Then in Section \ref{section_antenna_selection} we formulate the TAM problem. We consider the downlink of a single-cell MU-MIMO system. The BS which has $M$ antennas simultaneously serves $K= \{1,2,...,K\} \in \mathcal{K}, \lvert \mathcal{K} \rvert \leq M$ co-scheduled users out of $J \geq K$ users randomly distributed within the sector. The BS's cross-polarized antennas are arranged into $M_{col}$ vertical columns and $M_{row}$ horizontal rows where $M = 2 \times M_{col} \times M_{row}$. Each antenna element is individually controlled by a separate RF chain. Furthermore, we assume that both the BS and the users have full channel knowledge. The BS transmits $L = \sum_{k=1}^K L_k \leq M$ data streams where $k$th user receives its $L_k \leq N_{k}$ intended data streams at the $N_k$ receive antennas. The received signal $\mathbf{y}_k$ at the $k$th user is

\begin{equation}
\begin{aligned}
\label{eq:y_k_received}
    \mathbf{y}_k & =  \mathbf{H}_k \mathbf{W} \mathbf{x}  + \mathbf{n}_k , \\
    & =  \mathbf{H}_k \mathbf{W}_k \mathbf{x}_k +  \mathbf{H}_k \sum_{j=1, j \neq k}^K \mathbf{W}_j \mathbf{x}_j  +  \mathbf{n}_k, \\ 
\end{aligned}
\end{equation}

\noindent where $\mathbf{x} = [\mathbf{x}_1^T, \mathbf{x}_2^T, ..., \mathbf{x}_K^T]^T$ is the $L \times 1$ vector of transmitted symbols from the BS with $\mathbf{x}_k = [x_{k,1}, x_{k,2}, ..., x_{k, L_k}]^T$ intended for user $k$. We assume that the individual data streams have unit power, i.e. $\mathbb{E} [ \mathbf{x}_k \mathbf{x}_k^H ] = \mathbf{I}_{L_k} $. Let $\mathbf{W} = [\mathbf{W}_1, \mathbf{W}_2, ..., \mathbf{W}_K] \in \mathbb{C}^{M \times L}$ be the transmitter matrix at the BS for all the users where the transmitter matrix corresponding to user $k$ is $\mathbf{W}_k \in \mathbb{C}^{M \times L_k}$. The transmitter matrix is separated into a transmit power matrix and a beamforming matrix as $\mathbf{W} =  \mathbf{B} \mathbf{P}^{1/2}$ where $\mathbf{P} = diag([P_1, P_2,...,P_L]) \in \mathbb{C}^{L \times L} $ is the diagonal transmit power matrix where $P_l = P, \quad l=1,..,L$ is the power coefficient of the $l$th data stream considering an equal transmit power strategy. $\mathbf{B} = [\mathbf{B}_1, \mathbf{B}_2, ..., \mathbf{B}_K] \in \mathbb{C}^{M \times L}$ is the normalized beamforming matrix where the beamforming matrix for user $k$ being $\mathbf{B}_k \in \mathbb{C}^{M \times L_k}$. $\mathbf{H}_k \in \mathbb{C}^{N_{k} \times M} $ is the channel matrix from the BS to user $k$. The first term in \eqref{eq:y_k_received} is the intended signal component for user $k$, and the second term is the intra-cell interference from the other users in the cell. Finally, $\mathbf{n}_k $ is the $N_{k} \times 1$ noise vector for the user $k$ where  $\mathbf{n}_k \sim \mathcal{CN}(0,\mathbf{R}_{n_k})$ with $\mathbf{R}_{n_k} = \sigma_n^2 \mathbf{I}_{N_k}$. Note that $\mathbf{n}_k$ could also model the noise and the inter-cell interference in a multi-cell setup which we do not consider under the scope of this paper. The estimated transmitted data symbol $\mathbf{\hat{x}}_k$ for user $k$ is obtained as in \eqref{eq:x_k_estimated} by applying an appropriate receiver filter $\mathbf{V}_k \in \mathbb{C}^{L_k \times N_{k}} $ on the received signal $\mathbf{y}_k$ at each user $k$. 

\begin{equation}
\label{eq:x_k_estimated}
\begin{aligned}    
    \mathbf{\hat{x}}_k & = \mathbf{V}^H_k \mathbf{y}_k =  \mathbf{V}^H_k \mathbf{H}_k \mathbf{W} \mathbf{x}  + \mathbf{V}^H_k  \mathbf{n}_k , \\
    & =  \mathbf{V}^H_k \mathbf{H}_k \mathbf{W}_k \mathbf{x}_k \quad + \\
 & \qquad \mathbf{V}^H_k  \mathbf{H}_k \sum_{j=1, j \neq k}^K \mathbf{W}_j \mathbf{x}_j + \mathbf{V}^H_k \mathbf{n}_k. \\
\end{aligned}
\end{equation}

The signal-to-interference and noise-ratio (SINR) of user $k$ can be obtained by 

\begin{equation}
    \label{eq:sinr_basic}
    \text{SINR}_k = \frac{  | \mathbf{V}_k^H \mathbf{H}_k \mathbf{W}_k |^2} {\sum_{j=1,j \neq k}^K | \mathbf{V}_k^H \mathbf{H}_k \mathbf{W}_j |^2  + \sigma_n^2}.  
\end{equation}

\subsection{Receiver and Transmitter Matrix Design: MMSE Receiver and Eigen-Beamforming at the Transmitter}

At the receiver, the aim is to minimize the estimation error between the transmitted signal $\mathbf{x}_k$ and the estimated signal $\mathbf{\hat{x}}_k$. The error covariance matrix between each transmit and received symbol pair $\mathbf{x}_k$ and $\mathbf{\hat{x}}_k$ of user $k$ is defined as

\begin{equation}
\label{eq:error_cov_definition}
    \mathbf{E}_k =  \mathbb{E} \Big[(\mathbf{\hat{x}}_k - \mathbf{x}_k)(\mathbf{\hat{x}}_k - \mathbf{x}_k)^H \Big].
\end{equation}

Using \eqref{eq:x_k_estimated} and expanding the expression, 

\begin{equation}
\label{eq:error_cov}
\begin{aligned}
    \mathbf{E}_k &= \mathbf{V}_k^H \mathbb{E} \Big[ \mathbf{y}_k \mathbf{y}_k^H \Big] \mathbf{V}_k - \mathbf{V}_k^H \mathbb{E} \Big[ \mathbf{y}_k \mathbf{x}_k^H \Big] \\
    & \qquad \qquad \qquad - \mathbb{E} [ \mathbf{x}_k \mathbf{y}_k^H] \mathbf{V}_k   + \mathbf{I}_{L_k}, \\
    &= \mathbf{V}_k \Big( \mathbf{H}_k \mathbf{W} \mathbf{W}^H \mathbf{H}_k^H + \mathbf{R}_{n_k} \Big) \mathbf{V}_k^H - \mathbf{V}_k^H \mathbf{H}_k \mathbf{W}_k \\
    & \qquad \qquad  \qquad - \mathbf{W}_k^H \mathbf{H}_k^H \mathbf{V}_k^H + \mathbf{I}_{L_k}.
\end{aligned}
\end{equation}

The mean squared error (MSE) of the estimations corresponds to the diagonal entries in the error covariance matrix $ \mathbf{E}_k$, thus for a given transmitter matrix $\mathbf{W}_k$, the minimum mean square error (MMSE) receiver can be obtained as \cite{1312557_mu_mimo_mmse}

\begin{equation}
\begin{aligned}
\label{eq:mmse_receiver}
    \mathbf{V}_{k,mmse} & = \Big( \mathbf{H}_k \mathbf{W}  \mathbf{W} \mathbf{H}_k^H + \mathbf{R}_{n_k} \Big)^{-1}  \mathbf{H}_k \mathbf{W}_k,\\
    & = \Big( \mathbf{H}_k \mathbf{W}_k  \mathbf{W}_k^H \mathbf{H}_k^H + \mathbf{R}_{(I+n)_k} \Big)^{-1}  \mathbf{H}_k \mathbf{W}_k,\\
\end{aligned}
\end{equation}

\noindent with

\begin{equation}
\label{eq:if+n_cov}
\mathbf{R}_{(I+n)_k} = \mathbf{R}_{n_k} + \sum_{j=1, j \neq k}^K \mathbf{H}_k \mathbf{W}_j \mathbf{W}_j^H \mathbf{H}_k^H, \\
\end{equation}

\noindent where $\mathbf{R}_{(I+n)_k}$ can be defined as the interference plus noise covariance matrix for user $k$. Given this MMSE receiver, the original error covariance matrix in \eqref{eq:error_cov} can be simplified as 


\begin{align}
    \mathbf{E}_{k} & = \mathbf{I}_{L_k} - \mathbf{W}_k^H \mathbf{H}_k^H \Big ( \mathbf{W}_k^H \mathbf{H}_k^H \mathbf{H}_k \mathbf{W}_k \\ 
    & \hspace{10mm} + \mathbf{R}_{(I+n)_k} \Big)^{-1}\mathbf{H}_k \mathbf{W}_k, \\    
    \label{eq:error_cov_final_1}
    &= \Big( \mathbf{I}_{L_k} + \mathbf{W}_k^H \mathbf{H}_k^H  \mathbf{R}_{(I+n)_k}^{-1} \mathbf{H}_k \mathbf{W}_k \Big)^{-1}, \\
    \label{eq:error_cov_final}
    & = \Big( \mathbf{I}_{L_k} +\mathbf{H}_{eff,k}^H  \mathbf{R}_{(I+n)_k}^{-1} \mathbf{H}_{eff,k} \Big)^{-1},   
\end{align}

\noindent where we denote $\mathbf{H}_{eff,k} = \mathbf{H}_k \mathbf{W}_k$ as the effective channel of user $k$. Note that \eqref{eq:error_cov_final_1} follows from the matrix inversion Lemma \cite{1312557_mu_mimo_mmse}. The MSE of the ($k$th, $i$th) data stream is the $i$th diagonal element of $\mathbf{E}_{k}$, i.e. $\text{MSE}_{k,i} = [\mathbf{E}_{k}]_{ii}$. Then, using the inverse relationship between the MSE and SINR \cite{4355332_mmse_transceiver}, the SINR of user $k$ can be obtained as follows by taking the average MSE over the $L_k$ data streams,

\begin{equation}
    \label{eq:sinr}
    \text{SINR}_k = \frac{1}{L_k} \sum_{i=1}^{L_k}\Big(\frac{1}{\text{MSE}_{k,i}} - 1 \Big).
\end{equation}

The per-user rate for a given channel matrix $\mathbf{H}_k$ and transmitter matrix $\mathbf{W}_k$ can be then calculated by the usual rate equation for a given bandwidth $B$ as 

\begin{equation}
    \label{eq:user_rates}
    \text{r}_k (\mathbf{H}_k, \mathbf{W}_k) = B \hspace{2mm} \text{log}_2 \big(1 + \text{SINR}_k \big).
\end{equation}

For the transmitter matrix design, the aim is to reduce the interference or minimize the $\mathbf{R}_{(I+n)_k}$ in \eqref{eq:if+n_cov}. In \cite{Väisänen2018_eig_su_mimo}, Eigen-beamforming is shown to achieve the optimal channel capacity for the SU-MIMO case. The same approach could be extended for MU-MIMO case by considering each user’s channel independently and computing Eigen-beamforming coefficients using the SU-MIMO solution. The problem with this approach is the correlation between the beamforming coefficients of different users caused by closely located users which causes interference. It could be overcome by utilizing a scheduler that schedules users with near-orthogonal channels so that inter-user interference could be neglected in SINR estimation when using Eigen-beamforming \cite{Väisänen2018_eig_su_mimo}. Such a  semi-orthogonal user selection algorithm is proposed in \cite{1603708_sus_scheduling}. For this work, we have utilized a proportional fairness scheduler based on Eigen-beam cross-correlation. Each time when a candidate user is considered for scheduling, its Eigen-beam cross-correlation with the already selected users is calculated, and the user is added to the scheduled set only if it has a cumulative cross-correlation less than a predefined threshold. Thus, when the threshold value is set to be low enough and by having to select $K$ users among geographically distributed $J \geq K$ users, we can assume that the scheduler results in selecting users with minimum possible inter-user interference which can be negligible. Therefore, for the SINR estimations in TAM algorithms, we assume no inter-user interference due to the scheduler.

Let $\mathbf{R}_k = \mathbb{E}[\mathbf{H}_k^H \mathbf{H}_k ] \in \mathbb{C}^{M \times M}$ be the channel covariance matrix of user $k$, and $\mathbf{R}_{k,avg} = (\mathbf{R}_{k,H\_pol} + \mathbf{R}_{k,V\_pol})/2 \in \mathbb{C}^{M/2 \times M/2} $ be the channel covariance matrix of user $k$ be the average channel covariance matrix averaged across the two polarizations. Note that $\mathbf{R}_{k,H\_pol}$ and $\mathbf{R}_{k,V\_pol}$ are the two per-polarization covariance matrices extracted from $\mathbf{R}_k$. Then, the Eigen-beamforming is obtained as the eigenvector corresponding to the largest eigenvalue of the average channel covariance matrix $\mathbf{R}_{k, avg}$ by performing eigenvalue decomposition as,

\begin{equation}
    \label{eq:eig_dec}
    \mathbf{R}_{k, avg} = \mathbf{U}_k \mathbf{\Lambda}_k \mathbf{U}_k^H,
\end{equation}

\noindent where $\mathbf{\Lambda}_k = diag ([\lambda_1, \lambda_2, \dots])$ is the diagonal matrix consisting of the eigenvalues of $\mathbf{R}_{k, avg}$ and $\mathbf{U}_k = [\mathbf{u}_{k1}, \mathbf{u}_{k2},...]$ is the corresponding eigenvector matrix of size $ M/2 \times M/2$. Then, we can obtain the beamforming matrix as $\mathbf{B}_{k,eig}  = \mathbf{I}_{L_k} \otimes \mathbf{u}_{k1}$ where $\mathbf{u}_{k1}$ is the eigenvector corresponding to the largest eigenvalue $\lambda_1$ and $\otimes$ is the Kronecker product between two matrices. Then the transmitter matrix would be $\mathbf{W}_{k,eig} = \mathbf{B}_{k,eig} \mathbf{P}_k^{1/2} $ where $\mathbf{P}_k = diag([P_1,..,P_{L_k}]) \in \mathbb{C}^{L_k \times L_k} $ is the transmit power matrix of user $k$. We consider an equal transmit power strategy and thus $P_l = P, \quad l=1,..,L_k$.

The MMSE error covariance matrix $\mathbf{E}_{k,eig}$ can be obtained by \eqref{eq:error_cov_final_eigbf} by applying the Eigen-beamforming transmitter matrix in \eqref{eq:error_cov_final} as  $\mathbf{W}_k = \mathbf{W}_{k,eig}$. Note that we take $\mathbf{R}_{(I+n)_k} = \mathbf{R}_{n_k}$ since we assume that there is no inter-user interference due to the user scheduling algorithm implemented. Per-user SINR values and rates when using Eigen-beamforming in \eqref{eq:error_cov_final_eigbf} could be then obtained using \eqref{eq:sinr} and (\ref{eq:user_rates}).

\begin{equation}
\label{eq:error_cov_final_eigbf}
    \mathbf{E}_{k,eig} = \Big( \mathbf{I}_{L_k} + \mathbf{W}_{k,eig}^H \mathbf{H}_k^H  \mathbf{R}_{n_k}^{-1} \mathbf{H}_k \mathbf{W}_{k,eig} \Big)^{-1}.  
\end{equation}

\subsection{Zero-Forcing Beamforming at the BS} 
Here we also discuss the widely known zero-forcing beamformer which can be used at the BS to mitigate inter-user interference. Specifically, in an MU-MIMO setup where the $M$-antenna BS simultaneously transmits to $K \leq M$ single-antenna users where no user coordination is possible, zero inter-user interference in the cell ($\mathbf{h}_k \mathbf{w}_j = 0, j \neq k$) could be achieved by channel inversion, i.e. performing zero-forcing beamforming \cite{1207369_multiantenna_gbc_zf}, \cite{1261332_zeroforcing}. Let the system channel matrix is given by $\mathbf{H} = [\mathbf{h}_1^T, \mathbf{h}_2^T, ..., \mathbf{h}_K^T]^T \in \mathbb{C}^{K \times M} $. Note that here we denote the per-user channels as $\mathbf{h}_k \in \mathbb{C}^{1 \times M}$ as opposed to our previous notation $\mathbf{H}_k$ since they are row vectors when we consider single-antenna users, i.e. $N_k$=1. Then, the  zero-forcing beamforming matrix $\mathbf{W}_{zf} = [\mathbf{w}_1, \mathbf{w}_2, ..., \mathbf{w}_K] \in \mathbb{C}^{M \times K}$ with per-user beamforming vectors $\mathbf{w}_k \in \mathbb{C}^{M \times 1}$, can be obtained as \cite{1207369_multiantenna_gbc_zf}

\begin{equation}
    \label{eq:zf_beamforming}
    \mathbf{W}_{zf} = \mathbf{H}^{\dagger} = \mathbf{H}^H(\mathbf{H} \mathbf{H}^H)^{-1},
\end{equation}

\noindent where $\mathbf{H}^{\dagger}$ is the pseudo-inverse of the channel matrix. Then, with fixed power allocations at each user ($P_k = P, \quad \forall k$ ) and a normalized beamforming vector at the transmitter, the per-user rates for the zero-forcing transmitter can be obtained from 

\begin{equation}
\begin{aligned}
    \label{eq:sinr_zf_basic}
    \text{r}_{k,zf}(\mathbf{H}, \mathbf{W}_{zf}) & = B  \hspace{2mm} \text{log}_2 \Bigg( 1 +  \frac{ P_k  \frac{\lvert \mathbf{h}_k \mathbf{w}_k\rvert^2} {\lVert \mathbf{w}_k \rVert^2} } {\sum_{j=1,j \neq k}^K P_j \frac{\lvert \mathbf{h}_k \mathbf{w}_j\rvert^2} {\lVert \mathbf{w}_j \rVert^2}    +  \sigma_n^2} \Bigg), \\ 
    & = B \hspace{2mm} \text{log}_2 \Bigg( 1 + \frac{P \lvert \mathbf{h}_k \mathbf{w}_k\rvert^2} {\sigma_n^2 \lVert \mathbf{w}_k \rVert^2 }  \Bigg).
\end{aligned}
\end{equation}

This can be further reduced to \cite{1207369_multiantenna_gbc_zf}

\begin{equation}
    \label{eq:sinr_zf_simplified}
    r_{k,zf}(\mathbf{H}, \mathbf{W}_{zf}) = B \hspace{2mm} \text{log}_2  \Big( 1 + \frac{b_k P_k}{\sigma_n^2} \Big),
\end{equation}

\noindent where

\vspace{-2mm}

\begin{equation}
    \label{eq:sinr_zf_bk}
    b_k  = \frac{1}{\lVert   \textbf{w}_k \rVert^2} = \frac{1}{[(\mathbf{H} \mathbf{H}^H)^{-1}]_{k,k}}.
\end{equation}

\section{Transmit Antenna Muting for Power Minimization}
\label{section_antenna_selection}

In this section, we formulate the optimization problem for TAM. Our aim is to minimize the energy consumption of the cell by reducing the number of active antenna elements in the BS while guaranteeing the user QoS requirements. Our problem formulation is practical and results in a different policy of antenna muting than the rather commonly considered problem of maximizing the system sum rate or the energy efficiency per transmitted bit. 

In the BS's antenna panel, we assume that each antenna element is co-located with its cross-polarized counterpart, thus they demonstrate the same spatial correlation. Furthermore, we assume both co-located elements are controlled by the same switch.  
Therefore, given the antenna indices set $\mathcal{M}_{pol} = \{1,...,M/2 \}$, let $\mathbf{H}_{k, \mathcal{A}} = \mathbf{H}_k \mathbf{A}_\mathcal{A} \in \mathbb{C}^{ K \times M}$ be the channel matrix of user $k$, where $\mathcal{A} \subseteq \mathcal{M}_{pol}$ is the activate antenna element subset. Furthermore,  $\mathbf{A}_\mathcal{A} = \mathnormal{diag} ([\mathbf{a}^T_\mathcal{A}, \mathbf{a}^T_\mathcal{A}]) $ is the diagonal antenna activation matrix with  BS antenna activation vector $\mathbf{a}_\mathcal{A} = [a_1,....,a_M/2]^T \in \mathbb{C}^{M/2 \times 1}$ and the binary antenna element activation indicator is defined by $ a_i = \{1: i \in \mathcal{A},\; 0: i \notin \mathcal{A} \}$. Then, for any transmitter matrix $\mathbf{W}_{k, \mathcal{A}}$ given $\mathbf{H}_{k, \mathcal{A}}$, the error covariance matrices, per-user SINR values, and user rates can then be calculated using \eqref{eq:error_cov_final} - \eqref{eq:user_rates}.
We can then formulate the following optimization problem where the number of active antenna elements is minimized such that achievable per-user throughputs are higher than a given threshold, which captures the QoS guarantee.

\begin{equation}
\begin{aligned}
\label{eq:min_elements}
& \text{minimize}_{a_i \in \{0, 1\}} && \lVert \mathbf{a}_{\mathcal{A}} \rVert_{0}, \\
& \textrm{subject to} && r_{min} \leq r_k (\mathbf{H}_{k, \mathcal{A}}, \mathbf{W}_{k, \mathcal{A}}, \mathbf{a}_{\mathcal{A}}), \quad \forall k,\\
& && M_{min} \leq  \lVert \mathbf{a}_{\mathcal{A}} \rVert_{0},
\end{aligned}
\end{equation}

\noindent where $r_k (\mathbf{H}_{k, \mathcal{A}}, \mathbf{W}_{k, \mathcal{A}}, \mathbf{a}_{\mathcal{A}}) $
is the per-user rate for the given antenna activation $\mathbf{a}_\mathcal{A}$ obtained from \eqref{eq:error_cov_final} - \eqref{eq:user_rates}, and $r_{min}$ is the minimum allowed per-user rate to guarantee the QoS requirements. $M_{min}$ is the minimum number of active antennas (per-polarization) in the BS which is preset depending on the system parameters $M, K, N_k, L_k$.

Since \eqref{eq:min_elements} is combinatorial and non-convex, we cannot solve it easily using optimization techniques. It is an NP-hard problem that scales exponentially in processing run-time with the number of antennas. Heuristic-based algorithms could be used to solve it sub-optimally, however, such sub-optimal methods might not be feasible for practical implementation given their iterative nature and the processing complexity involved, which limits achieving proper energy savings from TAM. We address this computational complexity issue by developing an ML-based TAM algorithm that solves (\ref{eq:min_elements}) in a data-driven manner with lower computational complexity. The proposed NAM approach is detailed in Section \ref{section_ml}. Then, in Section \ref{section_heuristics} we present several heuristic-based algorithms to solve (\ref{eq:min_elements}) and compare their performance and the computational complexity with the proposed NAM approach.

\section{NAM: Neural Network-based TAM}
\label{section_ml}

In this section, we present NAM, an ML-based TAM approach capable of learning the TAM to solve (\ref{eq:min_elements}), which is the main contribution of this paper. We propose a supervised learning algorithm that can learn to approximate the labels produced by a baseline algorithm. Specifically, we train a NN to learn the antenna muting as a classification task, using the fixed column TAM method (described in Section \ref{section_heuristics_B}) as the baseline to generate labels. The design and implementation details of the NAM algorithm are detailed in the following sub-sections.

\subsection{NN Model Architecture for NAM}

The aim of the supervised learning approach for NAM is to learn the input-output mapping of the labels during model training. Selecting a proper input to the model is therefore crucial when designing a NN to learn the optimal antenna selections. The primary determining factor for the antenna muting is undoubtedly the channels between the BS and the users, along with the channel covariance matrices and beamforming vectors which are used for user throughput calculations. Furthermore, we need to consider the input dimensions as well, since the original channel matrices and channel covariance matrices have higher dimensions which will result in a complex model. 

After some initial investigations, the combination of two different inputs was found to give the desired performance in the NN and a feasible model complexity. The first input would be the averaged channel vectors of each co-scheduled user, averaged over all the sub-carriers and user antennas to reduce the model dimensionality. The second input would be the beamforming vectors of the co-scheduled users which are computed considering the full antenna array configuration. In each of the inputs, the real and imaginary parts are separated and normalized before taking them into the model. All the inputs from co-scheduled users are stacked together to form the model input $\mathbf{X} = [\mathbf{X}_1, \mathbf{X}_2, .., \mathbf{X}_K] \in \mathbb{R}^{M/2 \times 4 \times K}$, where,

\begin{equation}
\label{eq:dnn_inputs}
    \mathbf{X}_k = [\Re (\mathbf{W}_{k,pol}), \Im (\mathbf{W}_{k,pol}), \Re (\mathbf{H}_{k,pol}), \Im  (\mathbf{H}_{k,pol})],
\end{equation}

\noindent where $\mathbf{H}_{k,pol} \in \mathbb{C}^{M/2 \times 1} $ is the channel vector of user $k$ averaged over the two polarizations and user antennas, and $\mathbf{W}_{k,pol} \in \mathbb{C}^{M/2 \times 1} $ is the transmitter vector for user $k$ averaged over the two polarizations. $\Re (.)$ and $\Im (.)$ denote the real and imaginary components of the vectors respectively.

We have selected a simple CNN model with a single convolutional layer followed by two dense layers as the model architecture considering the acceptable performance and lower model complexity. 
Since we are implementing TAM as a classification problem, the labels are converted to one-hot vectors $\mathbf{y} = \mathbf{1}_y \in \mathbb{R}^N $, i.e. an N-dimensional vector where the $y$th element of which is equal to one and zero otherwise. Here, $N$ refers to the number of array configuration classes, which is 8 in this case. Similarly, the model output $\mathbf{\hat{y}}$ is also a one-hot encoded vector. Accordingly, the output layer has the \textit{Softmax} activation and the hidden layers have \textit{ReLU} activations. Then we train the model using the \textit{categorical cross-entropy} loss function for multi-class classification, which is given by

\begin{equation}
    \label{eq:cross_entropy}
    l_{sym}(\mathbf{y},   \hat{\mathbf{y}}) = - \sum_{i=0}^{N-1} y_i \log{\hat{y}_i}.
\end{equation}

\subsection{Asymmetric Custom Loss Function to Ensure QoS}

Training the model to minimize the \textit{categorical cross-entropy} loss would push the model outputs to be as close as to the labels but 100\% performance cannot be guaranteed. This means that user throughput requirements in \eqref{eq:min_elements} cannot be always guaranteed by the NN outputs which is a challenge in NN implementation. However, due to the nature of the antenna configuration options that we have in this problem, we can satisfy the user throughput requirements (guarantee reliability) if the NN output class is similar to or higher than the label class. Taking this into consideration, we have modified the original \textit{categorical cross-entropy} loss function to push the NN outputs towards guaranteeing reliability. Thus, we have formulated a new loss function,

\begin{equation}
    \label{eq:custom_loss}
    loss(\mathbf{y}, \hat{\mathbf{y}}) = l_{sym}(\mathbf{y}, \hat{\mathbf{y}})+ \lambda  l_{asy}(y_{max},  \hat{y}_{max}),
\end{equation}

\noindent where the first part of the loss function is the original \textit{categorical cross-entropy} loss. Then, $l_{asy}(y_{max},   \hat{y}_{max})$ is the penalty term added to force the confusion matrix from the NN to be upper diagonal in order to guarantee the throughput constraints. The actual class label and NN output class are denoted by $y_{max} = argmax (\mathbf{y})$ and $\hat{y}_{max} = argmax (\mathbf{\hat{y}})$ respectively. The confusion matrix $C \in \mathbb{R}^{N \times N}$ between the actual labels $\mathbf{y}$ and NN outputs $\hat{\mathbf{y}}$ is defined as $c_{i,j} = Prob(\hat{y}_{max} = j | {y}_{max} = i), i,j=0,..,N-1$. The accuracy of the NN is calculated from the diagonal of the confusion matrix where $acc = \sum_{i=j} c_{i,j}$. An upper diagonal confusion matrix implies having $y_{max} \leq \hat{y}_{max}$, thus, guaranteeing the QoS requirements as explained earlier.

\begin{figure}[ht]
\centerline{\includegraphics [width=0.52\textwidth]{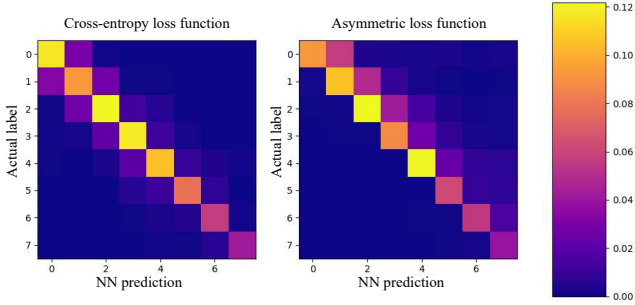}}
\caption{Confusion matrices for NN outputs obtained from the model trained with \textit{categorical cross-entropy} loss function in (\ref{eq:cross_entropy}) (left) and with new asymmetric loss function in (\ref{eq:custom_loss}) (right).}
\label{fig:confusion_matrix}
\end{figure}

Selecting a suitable function for $l_{asy}(y_{max},   \hat{y}_{max})$ and properly tuning the coefficient $\lambda$ is essential in order to improve the performance. We observed that first training the model with the original \textit{categorical cross-entropy} loss function and then retraining it with the above custom loss function gives a better improvement in QoS guarantee. We investigated different functions for the asymmetric loss and used

\begin{equation}
 \label{eq:custom_loss_2}
  l_{asy}(y_{max},   \hat{y}_{max}) =
    \begin{cases}
      {(y_{max}-\hat{y}_{max})^2}, \quad y_{max} > \hat{y}_{max},\\
      \alpha (y_{max}-\hat{y}_{max})^2, \quad \text{otherwise},
    \end{cases}       
\end{equation}

\noindent where (\ref{eq:custom_loss_2}) was implemented using the Keras \textit{LeakyReLU} layer. Furthermore, since the $argmax$ implementation in TensorFlow does not support backpropagation and gradient operations we replace it with the following approximation to obtain a differentiable output class.

\begin{equation}
 \label{eq:softargmax}
  {softargmax}(\hat{\mathbf{y}}) = \sum_i \frac{\exp({\beta \hat{y}_i})}{ \sum_j \exp({\beta \hat{y}_j})} i,    
\end{equation}

\noindent where $\hat{\mathbf{y}} = [\hat{y}_0, \hat{y}_1, ..., \hat{y}_{N-1}]$ is the NN output at the output from the \textit{Softmax} layer and $\beta \geq 1$. Thus, in (\ref{eq:custom_loss}) and (\ref{eq:custom_loss_2}) we use $y_{max}={softargmax}   (\mathbf{y})$ and $\hat{y}_{max} = {softargmax}   (\hat{\mathbf{y}})$ during the TensorFlow custom loss function implementation. Fig. \ref{fig:confusion_matrix} illustrates how the confusion matrix is changed when the outputs are obtained from the NN trained with the new asymmetric loss function in (\ref{eq:custom_loss}).

\begin{definition} \label{def:rel}
We quantify reliability $B(\alpha)$ as in the following: 
\begin{align}
B(\alpha) := \mathbb{E} \Big [ r_k (\mathbf{H}_{k, \mathcal{A}}, \mathbf{W}_{k, \mathcal{A}}, \mathbf{a}_{\mathcal{A}}, \alpha) - r_{min} \Big ],
\end{align}
where, $\alpha$ is introduced in (\ref{eq:custom_loss_2}), the expectation $\mathbb{E} \big[ \cdot \big]$ is taken over channel observations.
The rate $r_k (\mathbf{H}_{k, \mathcal{A}}, \mathbf{W}_{k, \mathcal{A}}, \mathbf{a}_{\mathcal{A}}, \alpha) $
is computed using antenna activation output
$\mathbf{a}_\mathcal{A}$ obtained from the NN that is trained by the asymmetric loss in (\ref{eq:custom_loss_2}) with parameter $\alpha$.
\end{definition}

\begin{definition}\label{def:monotone}
Let $f$ be a set function defined as $f: \mathcal{U} \rightarrow \mathbb{R}$, where $\mathcal{U} \subseteq \mathcal{O}$ is a subset of the domain set $\mathcal{O}$, then $f$ is said to be monotonic if: 
\begin{align}
    f(\mathcal{S}) \leq f(\mathcal{A}),  \quad \mathcal{S} \subseteq \mathcal{A}, \quad \mathcal{S},\mathcal{A} \subseteq \mathcal{U}.  
\end{align}
\end{definition}

\begin{proposition}\label{prop: 1}
If $\alpha <1 $ in (\ref{eq:custom_loss_2}), then the reliability according to the Definition \ref{def:rel}, $B(\alpha)$ is greater than or equal to that of the symmetric loss $\alpha=1$. More rigorously, we have: 
\begin{align}
    B(\alpha=1) \leq B(\alpha < 1).
\end{align}
\end{proposition}

\begin{proof}[Proof of Proposition \ref{prop: 1}]
We split the proof into two parts. At first, we show that when $\alpha < 1$ the asymmetric loss leads to the choice of more antennas, in expectation. In the second part, we prove that activating more antennas results in higher throughput using the monotonicity of the throughput function with respect to the cardinality of the selected antenna subset. We consider both zero-forcing beamforming and Eigen-beamforming for the rate monotonicity proof.

\begin{lemma}\label{lem:1}
For $\alpha<1$, the loss defined in (\ref{eq:custom_loss_2}), leads to more or equal antennas being chosen in expectation in comparison to $\alpha = 1$: 
\begin{align}
    \mathbb{E} \Big[\|\mathbf{a}^{\alpha=1}_{\mathcal{A}}\|_0 \Big] \leq \mathbb{E} \Big[\|\mathbf{a}^{\alpha < 1}_{\mathcal{A}}\|_0 \Big],
\end{align}

where $\|\mathbf{a}^{\alpha=1}_{\mathcal{A}}\|_0$ and $\|\mathbf{a}^{\alpha < 1}_{\mathcal{A}}\|_0$ are the number of active antennas selected by an NN which is trained using the loss from \eqref{eq:custom_loss_2} with $\alpha=1$ and $\alpha < 1$ respectively.
\end{lemma}

\begin{proof}[Proof of Lemma \ref{lem:1}]
Let's consider only the asymmetric part of the loss, which is given by \eqref{eq:custom_loss_2}. During the training, the gradient of the loss with respect to the weight parameters of the NN is computed. The gradient changes the parameters such that the prediction error decreases as much as possible. 
If $\alpha = 1$, then we have a symmetrical loss (MSE), i.e. the error of classification has no preference to select a larger subset of antennas than the label, or a smaller one.
However, in case of $\alpha<1$, if the predicted output $\hat{y}_{max}$ erroneously reports a higher number of antennas than the label, i.e. $\hat{y}_{max} > y_{max}$, due to the $\alpha<1$ we penalize this type of error less than the case where $\hat{y}_{max} < y_{max}$. Therefore, the weights of the NN are geared toward choosing more antennas than the label in case of an error, as it is a smaller increase in the loss. Thus,
\begin{align}
    \mathbb{E} \Big[\|\mathbf{a}^{\alpha=1}_{\mathcal{A}}\|_0 \Big] \leq \mathbb{E} \Big[\|\mathbf{a}^{\alpha < 1}_{\mathcal{A}}\|_0 \Big].
\end{align}
\end{proof}

\begin{lemma}
\label{lem:2}
The user rate function $r_k (\mathbf{H}_{\mathcal{A}}, \mathbf{W}_{zf, \mathcal{A}}, \mathbf{a}_{\mathcal{A}})$
is monotonic with respect to Definition \ref{def:monotone} when zero-forcing beamforming is used: 
\end{lemma}

\begin{proof}[Proof of Lemma \ref{lem:2}]
Given the antenna indices set $\mathcal{M} = \{1,2,...,M \}$, let $\mathbf{H}_\mathcal{S} = \mathbf{H} \mathbf{A}_\mathcal{S} \in \mathbb{C}^{ K \times M}$ be the channel matrix when the antenna element set $\mathcal{S} \subseteq \mathcal{M}$ is activated where $\mathbf{A}_\mathcal{S} = \mathnormal{diag} (\mathbf{a}_\mathcal{S}) $ is the diagonal antenna activation matrix with $\mathbf{a}_\mathcal{S} = [a_1,....,a_M]^T \in \mathbb{C}^{M \times 1}, \quad a_i = \{1: i \in \mathcal{S}, 0: i \notin \mathcal{S} \}$. Then, let $\mathbf{h}_m = [h_{1,m},..., h_{K,m}]^T \in \mathbb{C}^{ K \times 1}$ be the channel vector corresponding to an additional antenna element $m \in \mathcal{M}$, $m \notin \mathcal{S}$, and $\mathbf{H}_\mathcal{A} \in \mathbb{C}^{ K \times M}$ be the new channel matrix with the addition of new antenna $m$ where $\mathcal{A} = \mathcal{S} \cup \{ m \}$.

To show that the rate function is monotonic, we need to show that 

\begin{equation}
\label{eq:rate_inequality}
    r^\mathcal{S}_{k,zf}(\mathbf{H}_\mathcal{S}, \mathbf{W}_{zf,\mathcal{S}}, \mathbf{a}_{\mathcal{S}}) \leq  r^\mathcal{A}_{k,zf}(\mathbf{H}_\mathcal{A}, \mathbf{W}_{zf,\mathcal{A}}, \mathbf{a}_{\mathcal{A}}), \quad \forall k,
\end{equation}

\noindent where $r^\mathcal{S}_{k,zf}(\cdot)$ and $r^\mathcal{A}_{k,zf}(\cdot)$ are the rates of user $k$ when the set of antennas $\mathcal{S}$ and $\mathcal{A}$ are selected respectively. According to \eqref{eq:sinr_zf_basic} - \eqref{eq:sinr_zf_bk} it is clear that the user rates are proportional to $b_k = \frac{1}{[(\mathbf{H} \mathbf{H}^H)^{-1}]_{k,k}}$. Thus, proving \eqref{eq:rate_inequality} is equivalent to

\begin{equation}
\label{eq:HH_inequality}
\begin{aligned}
\frac{1}{[(\mathbf{H}_\mathcal{S} \mathbf{H}_\mathcal{S}^H)^{-1}]_{k,k}} & \leq \frac{1}{[(\mathbf{H}_\mathcal{A} \mathbf{H}_\mathcal{A}^H)^{-1}]_{k,k}}, \quad \forall k,
\end{aligned}
\end{equation}
which could be further simplified to: 
\begin{align}\label{eq:HH_inequality1}
[\mathbf{G}_\mathcal{A}^{-1}]_{k,k} & \leq [\mathbf{G}_\mathcal{S}^{-1}]_{k,k}, 
\end{align}

\noindent where $\mathbf{G}_\mathcal{S} := \mathbf{H}_\mathcal{S} \mathbf{H}_\mathcal{S}^H$ and $\mathbf{G}_\mathcal{A} := \mathbf{H}_\mathcal{A} \mathbf{H}_\mathcal{A}^H$ and $ \mathbf{G}_\mathcal{A} =  \mathbf{G}_\mathcal{S} + \mathbf{h}_m \mathbf{h}_m^H$.

From Sherman–Morrison formula \cite{1177729959_sherman_morrison}, \cite{1177729893_sherman_morrison_1}, the relationship between matrix inverses of two invertible matrices $\mathbf{A}, \mathbf{B} \in \mathbb{C}^{n \times n}$ which has the format $\mathbf{B} = \mathbf{A} + \mathbf{a}\mathbf{b}^H$ for some column vectors $\mathbf{a}, \mathbf{b} \in \mathbb{C}^{n}$ under the condition $(1 + \mathbf{b}^H\mathbf{A}^{-1}  \mathbf{a}) \neq 0$ is given by

\begin{equation}
    \label{eq:sherman_formula}
    \mathbf{B}^{-1} = (\mathbf{A} + \mathbf{a}\mathbf{b}^H )^{-1} =  \mathbf{A}^{-1} - \frac{\mathbf{A}^{-1} \mathbf{a} \mathbf{b}^H \mathbf{A}^{-1}}{ 1 + \mathbf{b}^H\mathbf{A}^{-1}  \mathbf{a}}.
\end{equation}

Using \eqref{eq:sherman_formula} above, we can obtain the following relationship between $\mathbf{G}_\mathcal{S}^{-1}$ and $\mathbf{G}_\mathcal{A}^{-1}$ as,

\begin{equation}
\label{eq:matrix_inverse_proof0}
\begin{aligned}
   &  \mathbf{G}_\mathcal{A}^{-1} = 
    \mathbf{G}_\mathcal{S}^{-1} -  \frac{\mathbf{G}_\mathcal{S}^{-1} \mathbf{h}_m \mathbf{h}_m^H \mathbf{G}_\mathcal{S}^{-1}}{ 1 + \mathbf{h}_m^H \mathbf{G}_\mathcal{S}^{-1}  \mathbf{h}_m}, \\
& \mathbf{G}_\mathcal{S}^{-1} - \mathbf{G}_\mathcal{A}^{-1}= \frac{\mathbf{G} \mathbf{G}^H}{ 1 + \mathbf{h}_m^H \mathbf{G}_\mathcal{S}^{-1}\mathbf{h}_m},
\end{aligned}
\end{equation}

\noindent where, $\mathbf{G}:= \mathbf{G}_\mathcal{S}^{-1} \mathbf{h}_m \in \mathbb{C}^{K \times K}$, and $\mathbf{G}\mathbf{G}^H$ is a positive semi-definite (PSD) matrix by construct. The denominator $ 1 + \mathbf{h}_m^H\mathbf{G}_\mathcal{S}^{-1} \mathbf{h}_m > 0$, since $\mathbf{G}_\mathcal{S}^{-1}$ is a PSD matrix. The diagonal elements corresponding to user $k=1,2,..,K$ can be presented from \eqref{eq:matrix_inverse_proof0} as

\begin{align} \label{eq:matrix_inverse_proof}
0\leq \frac{[\mathbf{G} \mathbf{G}^H]_{k,k}}{ 1 + \mathbf{h}_m^H \mathbf{G}_\mathcal{S}^{-1}\mathbf{h}_m} &= [\mathbf{G}_\mathcal{S}^{-1}]_{k,k} - [\mathbf{G}_\mathcal{A}^{-1}]_{k,k}, \;\;\forall k,\\
\label{eq:matrix_inverse_proof1}
 [\mathbf{G}_\mathcal{A}^{-1}]_{k,k}  & \leq  [\mathbf{G}_\mathcal{S}^{-1}]_{k,k}, \quad \forall k,
\end{align}
\noindent where the inequality in \eqref{eq:matrix_inverse_proof} follows from $\mathbf{G}^{-1}_\mathcal{S}$ and $\mathbf{G}$ are both PSD matrices and the diagonal elements of a PSD matrix are non-negative. 
In \eqref{eq:matrix_inverse_proof1} we obtain the claim \eqref{eq:HH_inequality1}, and the proof is complete. Thus, when using zero-forcing beamforming, the function user rate is monotonic with respect to Definition \ref{def:monotone}.

\end{proof}

\begin{lemma}
\label{lem:3}
The user rate function $r_k (\mathbf{H}_{k, eig, \mathcal{A}}, \mathbf{W}_{k, eig, \mathcal{A}}, \mathbf{a}_{\mathcal{A}})$
is monotonic with respect to Definition \ref{def:monotone} when Eigen-beamforming is used: 
\end{lemma}

\begin{proof}[Proof of Lemma \ref{lem:3}]
Let $\mathbf{W}_{k,eig,\mathcal{S}}, \mathbf{W}_{k,eig,\mathcal{A}} \in \mathbb{C}^{M \times L_k}$ be the Eigen-beamforming transmitter matrices of user $k$ when antenna element sets $\mathcal{S}$ and $\mathcal{A}$ are selected respectively. Without loss of generality, let's assume each user has a single data stream (i.e., $L_k = 1, \quad \forall k$). The corresponding per-user channel matrices are denoted by $\mathbf{H}_{k,\mathcal{S}}, \mathbf{H}_{k,\mathcal{A}} \in \mathbb{C}^{N_k \times M}$. From \eqref{eq:sinr} and \eqref{eq:error_cov_final_eigbf}, per-user SINR when using antenna element sets $\mathcal{S}$ and $\mathcal{A}$ (defined previously) should hold the following condition to satisfy the rate monotonicity.

\begin{equation}
\label{eq:rate_inequality_eigbf}
\begin{aligned}
    & r^\mathcal{S}_{k,eig}(\mathbf{H}_{k,\mathcal{S}}, \mathbf{W}_{k,eig,\mathcal{S}}, \mathbf{a}_{\mathcal{S}}) \leq  \\
    & \qquad \qquad \qquad r^\mathcal{A}_{k,eig}(\mathbf{H}_{k,\mathcal{A}}, \mathbf{W}_{k,eig,\mathcal{A}}, \mathbf{a}_{\mathcal{A}}), \quad \forall k, \\
    & \qquad \qquad \quad tr \Big( \mathbf{E}_{k,eig}^\mathcal{S} \Big) \geq  tr \Big( \mathbf{E}_{k,eig}^\mathcal{A} \Big). \\
\end{aligned}
\end{equation}

From the definition of eigenvalue decomposition, any eigenvalue $\lambda$ and its corresponding eigenvector $\mathbf{w}$ of the channel covariance matrix $\mathbf{R}_k = \mathbf{H}_k^H \mathbf{H}_k$ follows the condition $ \mathbf{R}_k \mathbf{w} = \lambda \mathbf{w}$. Also, the noise covariance matrix $\mathbf{R}_{n_k} = \sigma^2 \mathbf{I}$ and we consider equal power transmission for all the users.  Therefore, we can simplify the MSE expressions as, 

\begin{equation}
\label{eq:eigbf_proof_1}
\begin{aligned}
    & tr \big( \mathbf{E}_{k,eig}^\mathcal{S} \big) = tr \big( \big[  \mathbf{I}_{L_k} + \\
    & \hspace{20mm} \mathbf{W}_{k,eig,\mathcal{S}}^H \mathbf{H}_{k,\mathcal{S}}^H  \mathbf{R}_{n_k}^{-1} \mathbf{H}_{k,\mathcal{S}} \mathbf{W}_{k,eig,\mathcal{S}} \big]^{-1} \big), \\
    & = tr \Big( \big[  \mathbf{I}_{L_k} + 1/\sigma^2 (\lambda_{max,\mathcal{S}} \mathbf{W}_{k,eig,\mathcal{S}}^H \mathbf{W}_{k,eig,\mathcal{S}}) \big]^{-1} \Big), \\
    & = \frac{1}{\big(  1 + \frac{P \lambda_{max,\mathcal{S}}}{\sigma^2}   \big)}, \\
\end{aligned}
\end{equation}

\vspace{-1mm}

\begin{equation}
\label{eq:eigbf_proof_2}
\begin{aligned}
    & tr \big( \mathbf{E}_{k,eig}^\mathcal{A} \big) = tr \big( \big[  \mathbf{I}_{L_k} + \\
    & \hspace{20mm} \mathbf{W}_{k,eig,\mathcal{A}}^H \mathbf{H}_{k,\mathcal{A}}^H  \mathbf{R}_{n_k}^{-1} \mathbf{H}_{k,\mathcal{A}} \mathbf{W}_{k,eig,\mathcal{A}} \big]^{-1} \big), \\
    & = tr \Big( \big[ \mathbf{I}_{L_k} + 1/\sigma^2 (\lambda_{max,\mathcal{A}} \mathbf{W}_{k,eig,\mathcal{A}}^H \mathbf{W}_{k,eig,\mathcal{A}}) \big]^{-1} \Big), \\
    & = \frac{1}{\big(  1 + \frac{P \lambda_{max,\mathcal{A}}}{\sigma^2} \big)}   , \\
\end{aligned}
\end{equation}

\noindent where $\lambda_{max,\mathcal{S}}, \lambda_{max,\mathcal{A}}$ are the maximum eigenvalues of the channel covariance matrices $\mathbf{R}_{k,\mathcal{S}} = \mathbf{H}_{k,\mathcal{S}}^H \mathbf{H}_{k,\mathcal{S}}$ and $\mathbf{R}_{k,\mathcal{A}} = \mathbf{H}_{k,\mathcal{A}}^H \mathbf{H}_{k,\mathcal{A}}$ and $ \mathbf{W}_{k,eig,\mathcal{S}}$ and $ \mathbf{W}_{k,eig,\mathcal{A}}$ are the Eigen-beamforming transmitter matrices corresponding to those maximum eigenvalues. The Cauchy interlace theorem states that the eigenvalues $\lambda_1 \leq ... \leq \lambda_n$ of a Hermitian matrix of order $n$ are interlaced with the eigenvalues $\mu_1 \leq ... \leq \mu_{(n-1)}$
of a principal sub-matrix of order $(n-1)$, i.e., $\lambda_{k} \leq \mu_{k} \leq \lambda_{k+1}, k=1,...,(n-1)$ \cite{10.2307/4145217_cauchy_interlace}. Thus, for the maximum eigenvalues of $\mathbf{R}_{k,\mathcal{S}}$ and $\mathbf{R}_{k,\mathcal{A}}$ we have $\lambda_{max,\mathcal{S}} \leq \lambda_{max,\mathcal{A}}$ \footnote{If we remove the zero-valued rows and columns (which arise from the inactive antenna elements) from the two covariance matrices $\mathbf{R}_{k,\mathcal{S}}$ and $\mathbf{R}_{k,\mathcal{A}}$, we can see that the resulting matrix $\mathbf{\tilde{R}}_{k,\mathcal{S}}$ is a principal sub-matrix of $\mathbf{\tilde{R}}_{k,\mathcal{A}}$. Then by applying the Cauchy interlace theorem to them we have $\lambda_{max,\mathcal{S}} \leq \lambda_{max,\mathcal{A}}$ for the maximum eigenvalues of $\mathbf{\tilde{R}}_{k,\mathcal{S}}$ and $\mathbf{\tilde{R}}_{k,\mathcal{A}}$. Since the additional zero-valued columns and rows don't affect the eigenvalue decomposition of $\mathbf{R}_{k,\mathcal{S}}$ and $\mathbf{R}_{k,\mathcal{A}}$, they will also have the same set of eigenvalues as those of $\mathbf{\tilde{R}}_{k,\mathcal{S}}$ and $\mathbf{\tilde{R}}_{k,\mathcal{A}}$ respectively, and the previous result $\lambda_{max,\mathcal{S}} \leq \lambda_{max,\mathcal{A}}$ will still hold.}. Therefore, from \eqref{eq:eigbf_proof_1} and (\ref{eq:eigbf_proof_2}) we have,

\begin{equation}
\label{eq:eigbf_proof_3}
\begin{aligned}
     \frac{\sigma^2}{\big(  \sigma^2 + P \lambda_{max,\mathcal{S}}  \big)} & \geq \frac{\sigma^2}{\big(  \sigma^2 + P \lambda_{max,\mathcal{A}}  \big)}, \\
    tr \Big( \mathbf{E}_{k,eig}^\mathcal{S} \Big) & \geq tr \Big( \mathbf{E}_{k,eig}^\mathcal{A} \Big). 
\end{aligned}
\end{equation}

Thus, \eqref{eq:eigbf_proof_3} satisfies the requirement in \eqref{eq:rate_inequality_eigbf}, and therefore it is proved that when using Eigen-beamforming, the function user rate is monotonic with respect to Definition \ref{def:monotone}. \\
\vspace{-2mm}
\end{proof}

\begin{figure}[ht]
\centerline{\includegraphics [width=0.48\textwidth]{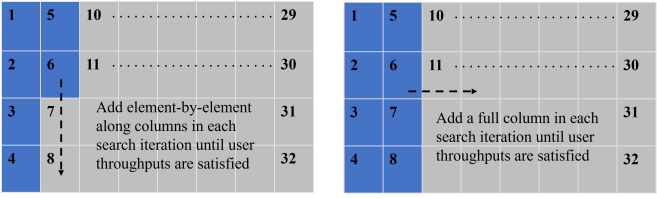}}
\caption{(a). Sequential antenna muting method.  (b). Fixed column antenna muting method.}
\label{fig:antenna_config}
\end{figure}

According Lemma \ref{lem:1}, \ref{lem:2} and \ref{lem:3} we have
\begin{align}
    &\mathbb{E} \Big[\|\mathbf{a}^{\alpha=1}_{\mathcal{A}}\|_0 \Big] \leq \mathbb{E} \Big[\|\mathbf{a}^{\alpha < 1}_{\mathcal{A}}\|_0 \Big], \\
    &\mathbb{E} \Big [ r_k (\mathbf{H}_{k, \mathcal{A}}, \mathbf{W}_{k, \mathcal{A}}, \mathbf{a}_{\mathcal{A}}, \alpha=1) \Big ] \leq \nonumber \\ 
    &\hspace{2cm} \mathbb{E} \Big [ r_k (\mathbf{H}_{k, \mathcal{A}}, \mathbf{W}_{k, \mathcal{A}}, \mathbf{a}_{\mathcal{A}}, \alpha<1) \Big ],\\
    &\mathbb{E} \Big [ r_k (\mathbf{H}_{k, \mathcal{A}}, \mathbf{W}_{k, \mathcal{A}}, \mathbf{a}_{\mathcal{A}}, \alpha=1) - r_{min}\Big ] \leq \nonumber \\ 
    &\hspace{2cm}  \mathbb{E} \Big [ r_k (\mathbf{H}_{k, \mathcal{A}}, \mathbf{W}_{k, \mathcal{A}}, \mathbf{a}_{\mathcal{A}}, \alpha<1) -r_{min}\Big ].
\end{align}

The proof of Proposition \ref{prop: 1} is completed.

\end{proof}

\section{Heuristic TAM}
\label{section_heuristics}

In this section, we present the newly proposed greedy TAM algorithm and two other low-complexity heuristic algorithms for solving the optimization problem in \eqref{eq:min_elements}. The ``fixed column antenna muting" algorithm proposed in Section \ref{section_heuristics_B} is used to generate labels for training the NN in NAM proposed in Section \ref{section_heuristics_B}.

\subsection{Greedy TAM Algorithm}

The greedy TAM algorithm which we propose here iteratively selects the best antenna element which gives the best sum rate performance until all the per-user throughput requirements are satisfied. It starts from a single element that gives the best rate for the co-scheduled users and iteratively adds elements until throughput conditions and the minimum number of antennas are satisfied. For the given iteration, if there are multiple antenna configurations that satisfy throughput conditions, the antenna configuration with the maximum sum rate is selected. If throughput conditions and the minimum number of antennas are not met for the current number of active antenna elements, the antenna element which gives the best sum rate for the co-scheduled users is added to the active antenna element set and element addition is continued. The algorithm repeats this process until both constraints are satisfied. The main steps of the greedy TAM algorithm are summarized in Algorithm 1. It is not an exhaustive search method since it does not search through all the different combinations in the search space, but it is a greedy approach in the sense each time it selects the best antenna element from the subset of remaining elements. Note that this approach allows minimizing the number of active antenna elements as well as optimizing the array shape.

\begin{algorithm}
\caption{Greedy search-based TAM}
\label{alg:antenna_sel}
\begin{algorithmic}
\REQUIRE $\text{channel matrices } \mathbf{H}_k, \quad k = 1,2,...,K $  
\STATE $\mathbf{a\_idc} = [1,2,…,M/2] $
\STATE $\mathbf{a}_\mathcal{A} = \mathbf{0}, \quad \mathbf{a}_\mathcal{A}= [a_1, a_2,..., a_{M/2}], a_i \in \{0, 1\}$

\FOR {$M_i=1,2,\ldots,M/2$}
    \STATE $\mathbf{a\_selected} = [  ]$
    \STATE $\mathbf{remaining\_idc} = diff(\mathbf{a\_idc}, \mathbf{a\_selected})$
    
    \FOR {$i \in \mathbf{remaining\_idc} $}        
    \STATE $\mathbf{a}_\mathcal{I} = [\mathbf{a\_selected}, i]$, $\mathbf{A}_\mathcal{I} = diag([\mathbf{a_\mathcal{I}}, \mathbf{a_\mathcal{I}}]^T)$
    \STATE obtain $\mathbf{H}_{k,\mathcal{I}} = \mathbf{H}_k \mathbf{A}_\mathcal{I}$
    \STATE calculate $\mathbf{R}_{k,avg,\mathcal{I}}, \hspace{2mm} \mathbf{W}_{k,\mathcal{I}}, \hspace{2mm} r_{k,\mathcal{I}}, \hspace{5mm} \quad k = 1,...,K$  
    
    \IF {$r_{min} \leq r_{k,\mathcal{I}} (\mathbf{H}_{k,\mathcal{I}},\mathbf{W}_{k,\mathcal{I}}, \mathbf{a}_\mathcal{I}) \hspace{3mm} \forall k$}
            
        \STATE \text{collect to set } $\mathcal{P}$: 
        \STATE \hspace{2mm} $\text{index}  \hspace{2mm}  i$,   $\mathbf{W}_{k,\mathcal{I}},   r_{k,\mathcal{I}}(\mathbf{H}_{k,\mathcal{I}},\mathbf{W}_{k,\mathcal{I}}, \mathbf{a}_\mathcal{I})$
            
    \ELSE
        \STATE \text{collect to set } $\mathcal{Q}$: 
        \STATE \hspace{2mm} $\text{index}  \hspace{2mm}  i$,   $\mathbf{W}_{k,\mathcal{I}},   r_{k,\mathcal{I}}(\mathbf{H}_{k,\mathcal{I}},\mathbf{W}_{k,\mathcal{I}}, \mathbf{a}_\mathcal{I})$
                       
    \ENDIF
    \ENDFOR

    \IF {$ \mathcal{P} \neq \emptyset$} 
    \STATE $ i_{max} = argmax_j \Big(\sum_{k=1}^K r_{k,\mathcal{I},j}\Big), \hspace{2mm} j\in  \mathcal{P} $
    \STATE $\mathbf{a\_selected} = [\mathbf{a\_selected},  i_{max}]$
        \IF {$ M_{min} \leq M_{i}$} 
        \STATE $ \mathbf{W}_{k, \mathcal{A}} = \mathbf{W}_{k,\mathcal{I}, i_{max}} $ 
        \STATE $ \mathbf{H}_{k, \mathcal{A}} = \mathbf{H}_{k,\mathcal{I}, i_{max}} $ 
        \STATE break
        \ELSE
        \STATE \text{continue adding elements}
        \ENDIF
        
    \ELSE 
    \STATE $ i_{max} = argmax_j \Big(\sum_{k=1}^K r_{k,\mathcal{I},j}\Big), \hspace{2mm} j\in  \mathcal{Q} $
    \STATE $\mathbf{a\_selected} = [\mathbf{a\_selected},  i_{max}]$
    \STATE \text{continue adding elements}
        
    \ENDIF

  	\ENDFOR
  	
  	\STATE $ \text{antenna muting vector} \quad \mathbf{a}_\mathcal{A}[\mathbf{a\_selected}] = 1$
  	
    \STATE \text{reutrn}
    $\mathbf{a}_\mathcal{A}, \mathbf{A}_{\mathcal{A}}, \mathbf{H}_{k,\mathcal{A}},\mathbf{W}_{k,\mathcal{A}}, k=1,2,...,K $

\end{algorithmic}
\end{algorithm}

\begin{algorithm}
\caption{Sequential antenna muting}
\label{alg:antenna_sel_sequential}
\begin{algorithmic}
\REQUIRE $\text{channel matrices } \mathbf{H}_k, \quad k=1,2,...,K $ 
\STATE $\mathbf{a\_idc} = [1,2,…,M/2] $
\STATE $\mathbf{a}_\mathcal{A} = \mathbf{0}, \quad \mathbf{a}_\mathcal{A}= [a_1, a_2,..., a_{M/2}], a_i \in \{0, 1\}$
\FOR {$i=1,2,\ldots,M/2$}
    \STATE $\mathbf{a}_\mathcal{I} = \mathbf{a\_idc}[1:i]$
    $\mathbf{A}_\mathcal{I} = diag([\mathbf{a_\mathcal{I}}, \mathbf{a_\mathcal{I}}]^T)$
    \STATE obtain $\mathbf{H}_{k,\mathcal{I}} = \mathbf{H}_k \mathbf{A}_\mathcal{I}$
    \STATE calculate $\mathbf{R}_{k,avg,\mathcal{I}}, \hspace{2mm} \mathbf{W}_{k,\mathcal{I}}, \hspace{2mm} r_{k,\mathcal{I}}, \hspace{5mm} \quad k = 1,...,K$ 

    \IF {$(r_{min} \leq r_{k,\mathcal{I}}), \hspace{3mm} \forall k  \hspace{2mm} \mathbf{and}   \hspace{2mm} (M_{min} \leq i) $}
        \STATE break
    \ELSE 
        \STATE \text{continue adding elements}
    \ENDIF
\ENDFOR
  	\STATE $ \text{antenna muting vector} \quad \mathbf{a}_\mathcal{A}[\mathbf{a}_i] = 1$
    \STATE \text{return} $\mathbf{a}_\mathcal{A}, \mathbf{A}_{\mathcal{A}}, \mathbf{H}_{k,\mathcal{A}},\mathbf{W}_{k,\mathcal{A}}, k=1,2,...,K $
\end{algorithmic}
\end{algorithm}

\begin{table}[ht]
\caption{Simulation Parameters}
\begin{center}
\begin{tabular}{|l|l|}
\hline
Channel model & 3GPP 38.901 UMi 3D \cite{3gpp_36.873_channel_model} \\ \hline
Inter-site distance & 200 m \\ \hline
Num. of cells & One 120° degree sector\\ \hline
BS antenna configuration & $M = 64 $\\ 
$M =2\times N_{col} \times N_{row}$ &   $(2 \times 8 \times 4) $\\ \hline
Carrier frequency & 3.5 GHz \\ \hline
BS TX Power & 53 dBm \\ \hline
Transmission Bandwidth & 100 MHz (273 PRB)\\ \hline
Waveform & OFDM with 30 kHz subcarrier \\ 
  & spacing\\ \hline
Traffic model & Full buffer  \\ \hline
Number of user antennas  & $N_k = 4, \quad \forall k$  \\ \hline
Number of data streams  & $L_k = 2, \quad \forall k$  \\ \hline
User's receiver model  & Genie MMSE  \\ \hline
User distribution & Uniform random distribution \\  \hline
Avg. number of users  & $J = 10$ \\ 
per-sector  &  \\ \hline
Max. number of  & $K = 4$ \\ 
co-sched. users &  \\ \hline
Link adaptation & Perfect link adaptation based on  \\ 
& truncated Shannon (capped to 8 bits) \\ \hline
Subframe duration & 0.5 ms including 12 DL data symbols \\ 
&and 2 DL DMRS (pilot) symbols \\ \hline
UE speed & 3 km/h \\ \hline
Throughput threshold & $r_{min} = 0.3 Mbits/slot$ \\ \hline
Min. number of active  & $M_{min} = 4$ \\ 
 antennas (per-polarization) & \\ \hline
\end{tabular}
\end{center}
\label{tab:simulation}
\end{table}

\subsection{Heuristic-based TAM}
\label{section_heuristics_B}

Here we define two heuristics for TAM which have lower computational complexity than the greedy antenna muting algorithm. The first method is referred to as ``sequential antenna muting" which is also illustrated in Fig. \ref{fig:antenna_config} (a) where individual elements are added along the antenna panel until the per-user throughput requirements and the minimum number of active antenna requirements are satisfied. We have summarized the basic steps of the sequential antenna muting algorithm in Algorithm 2. It has a reduced complexity compared to greedy-based antenna muting since there is no search for the optimal antenna element in each iteration. A more simpler and straight-forward heuristic is to select an antenna configuration from a set of given fixed set of array configurations (i.e. fixed columns with 4,8,12,...,28,32 active elements per polarization, fixed rows with 4,8,12,...,28,32 active elements per polarization). We use this ``fixed column antenna muting" method which is illustrated in Fig. \ref{fig:antenna_config} (b) to generate the labels for the NN.

\section{Simulations and Results}

In this section, we present the numerical simulations to compare the performance of the proposed NAM algorithm and the heuristic TAM algorithms.

\subsection{Simulation Parameters}

\begin{table}[ht]
\caption{Classification accuracy of the proposed NAM approach}
\begin{center}
\begin{tabular}{|l|l|l|}
\hline
\textbf{Dataset} & \textbf{NN model}  & \textbf{NN model} \\
 & \textbf{trained with (\ref{eq:cross_entropy}) }  & \textbf{retrained with (\ref{eq:custom_loss})} \\ \hline
Training dataset & 88\%  & 83.9\%  \\ \hline
Test dataset & 83.1\%  & 80.9\% \\ \hline

\end{tabular}
\end{center}
\label{tab:dnn_performance_accuracy}
\end{table}

\begin{table}[ht]
\caption{QoS guarantee of the NAM approach where the NN output class is equal or higher than the label class confirming satisfying the QoS}
\begin{center}
\begin{tabular}{|l|l|l|}
\hline
\textbf{Dataset} & \textbf{NN model}  & \textbf{NN model} \\
 & \textbf{trained with (\ref{eq:cross_entropy}) }  & \textbf{retrained with (\ref{eq:custom_loss})} \\ \hline
Training dataset &  93.4\% & 97.3\% \\ \hline
Test dataset &  90.7\% & 95.0\% \\ \hline

\end{tabular}
\end{center}
\label{tab:dnn_performance_qos}
\end{table}

As discussed in Section \ref{section_system_model}, we have considered an MU-MIMO system with a single cell, covering a 120° degree sector. The important simulation parameters are summarized in Table \ref{tab:simulation}. We have used 3GPP's 3D Urban Micro (UMi) channel model \cite{3gpp_36.873_channel_model}. For user scheduling a proprietary algorithm is used where up to $K$ users out of the $J$ users are co-scheduled in each instance.

\subsection{Simulation Setup}

All the simulations including the data generation and antenna muting algorithms are Python-based implementations. The NN model implementation, training, and testing are done in TensorFlow \cite{tensorflow} using the Keras library. For the NN implementation, we have generated a dataset of 500 random user drops with 2000 instances (i.e., slots) in each drop (where each instance corresponds to 0.5 ms of time), resulting in 1 million data samples. It was partitioned into training, validation, and test sets with 80\%, 10\%, and 10\% of the total samples respectively. The training dataset is used for model training, whereas the validation dataset is used for tracking the model training performance and fixing the hyper-parameters. Finally, the model performance is evaluated on the test dataset. As mentioned earlier we used the fixed column antenna muting method to create labels that select one out of the $N=8$ given antenna configurations. We have trained the model for 100 epochs for a mini-batch size of 64 samples with Adam optimizer with a learning rate of 0.001. The performance of the different TAM algorithms discussed in Section \ref{section_heuristics} is evaluated on the above-mentioned test dataset in order to compare their performance against the proposed NAM method.

\subsection{Results and Discussion}

The classification accuracy of the proposed NAM algorithm is summarized in Table \ref{tab:dnn_performance_accuracy} for the training and test datasets. There, accuracy refers to when the NN output class is equal to the label class. Table \ref{tab:dnn_performance_qos} presents the reliability improvement brought by the newly introduced loss in (\ref{eq:custom_loss}), where the QoS guarantee is defined as when the NN output class is equal or higher than the label class which confirms satisfying the throughput constraints in (\ref{eq:min_elements}). We observed that retraining with the new loss function improves the QoS guarantee capability while slightly reducing the accuracy.

\begin{figure}[ht]
\centerline{\includegraphics [width=0.48\textwidth]{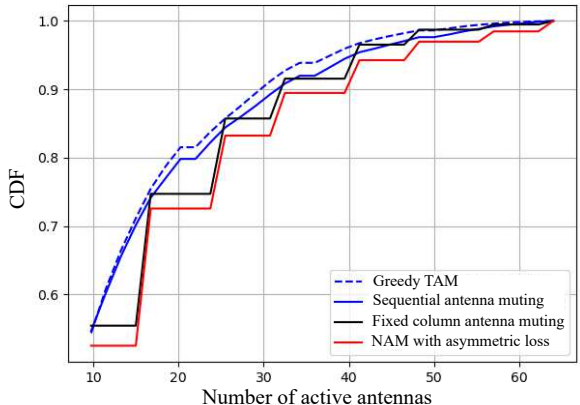}}
\caption{Cumulative distribution of active antenna numbers in different TAM strategies.}
\label{fig:antenna_selection_cdf}
\end{figure}

\begin{figure}[ht]
\centerline{\includegraphics [width=0.52\textwidth]{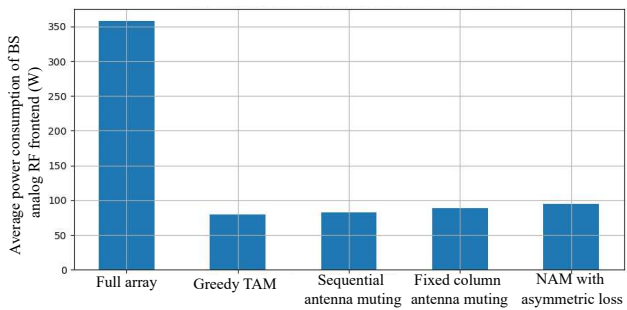}}
\caption{Average power consumption of the BS analog RF frontend module \cite{8644331_nokia} for different TAM approaches.}
\label{fig:avg_antenna_number}
\end{figure}

Fig. \ref{fig:antenna_selection_cdf} shows the cumulative distribution function (CDF) of active antenna numbers in different antenna muting algorithms. It can be observed that the NAM closely follows the curve of the fixed column antenna muting algorithm with only a slight shift. The sequential antenna muting algorithm has a slightly better performance due to the granularity it has in selecting the optimal number of antenna elements compared to the fixed column TAM. As expected, the greedy TAM has the best performance out of all the heuristics. In all methods, more than 50\% of the time only 8 antenna elements are needed to satisfy the user throughput requirements. While full array configuration has all the $M=64$ antennas active at all times, the greedy TAM, sequential antenna addition, and fixed column TAM method only require 14.22, 14.85, and 15.83 active antenna elements respectively on average to satisfy the given QoS requirements. The NAM closely follows its baseline by having 17.00 active antennas on average, which is approximately 73\% energy saving compared to the full array configuration.

\begin{figure}[ht]
\centerline{\includegraphics [width=0.54\textwidth]{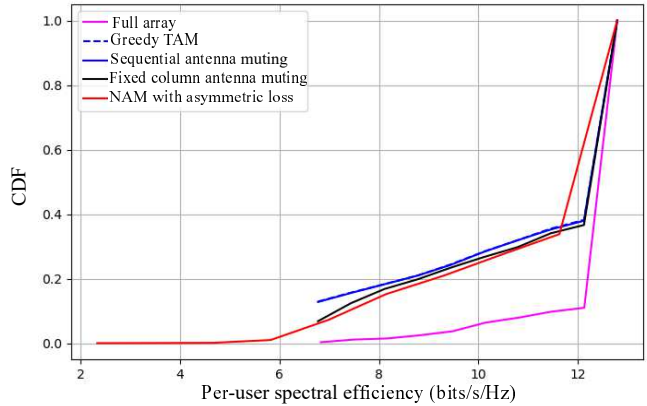}}
\caption{Cumulative distribution of per-user spectral efficiency for different TAM strategies.}
\label{fig:userrates_cdf}
\end{figure}

The estimated average power consumption at the BS analog RF frontend module is shown in Fig. \ref{fig:avg_antenna_number} for the different TAM methods discussed. We have used the power consumption values of the different analog RF frontend building blocks (TX conversion unit, TX power amplifier, RX conversion unit, and RX low noise amplifier) as given in \cite{8644331_nokia} in estimating the power consumption. It is clearly visible that TAM could achieve huge power savings at the RF frontend compared to the full array configuration.

Fig. \ref{fig:userrates_cdf} illustrates the CDFs of the achievable instantaneous per-user spectral efficiency (SE) of different TAM methods. For the simulations, we have used the user rate threshold as $r_{threshold} = 0.3 Mbits/slot $ which translates to $SE_{threshold} = 6.105 bit/s/Hz$. The heuristic-based methods have full reliability where per-user SE is always higher than $SE_{threshold}$, however, the NN has around 95\% of guaranteeing the user QoS requirements.

\section{Complexity Analysis}

In this section, we compare the computational complexity of the heuristic algorithms and the ML-based TAM in detail. First, we consider the main computations involved in the heuristic algorithms. Let $M_i$ be the number of active antenna elements in each search round of a given heuristic algorithm. For a single iteration with $M_i$ antenna elements, the calculations involve computing the channel covariance matrix, obtaining the Eigen-beamforming vector, and user rate calculation to check if the user rate requirements are satisfied. We have to perform all these steps $K$ times to do the rate calculation for all the scheduled users. Channel covariance matrix computation requires $\approx M_i^3/8 $ floating point operations (FPOs) complexity (assuming that the channel covariance matrix is obtained from the channel averaged over both polarizations). Computation of maximum eigenvalue and its corresponding eigenvector needs $\approx(M_i^3/8 + (M_i/2   log2   (M_i/2))   log b)$ FPOs, where $2^{-b}$ is the relative error bound accepted for the eigenvalues \cite{eig_complexity}, which we approximate as $(M_i^3/8)$ in the FPO calculations. Then the rate calculation step involves several matrix multiplications, additions, and inverse operations resulting in a total complexity of $\text{PRBs} \times (N_k M_i L_k + N_k^2 L_k + N_k L_k^2 + L_k^3)$ where $\text{PRBs}=273$ is the number of physical resource blocks (PRBs), $N_k$ is the number of user antennas, and $L_k=2$ corresponds to PDSCH rank 2 transmissions. Table \ref{tab:complexity} summarizes the required FPOs per channel instance in each of the three heuristics assuming a uniform distribution of $N$ antenna classes, and the FPO count for the NN. The total FPO count for a single iteration in heuristic algorithms which consists of the above-mentioned steps is taken as $F_i$ as given in \eqref{eq:flop_count_fi} below, which is used in calculations in Table \ref{tab:complexity}.

\begin{equation}
\begin{aligned}
\label{eq:flop_count_fi}
    F_i & = K \times \Big [ (M_i^3/8) + (M_i^3/8) + \\
    & \quad \text{PRBs} \times (N_k M_i L_k + N_k^2 L_k + N_k L_k^2 + L_k^3) \Big ]. \\
\end{aligned}
\end{equation}

In the fixed column antenna muting algorithm, the search for the optimal antenna configuration is done to find the optimal $M^* \in \{8,16,24,...,64 \}$ transmit antennas, starting from the lowest value and iterating over the 8 options until the per-user rate requirements are satisfied. Let $F_i \in \{F_0, F_{1},.., F_{7}\}$ be the total number of FPOs required when the active antenna number is $M_i \in \{8,16,24,...,64\}$ respectively. It should be noted that the number of iterations to find the correct antenna configuration $M^*$ is dependent on the user distribution and the channels between the BS and the users. For a fair estimation of FPO count, let us assume a uniform distribution of $N$ antenna classes, i.e., $P(i) = 1/N, \forall i$. The average number of FPOs required for a single instance is $\sum_{i=0}^{N-1} \sum_{j=0}^i P(j) F_i = \sum_{i=0}^{N-1} (i/N ) F_i$.

We can also obtain the average number of FPOs required for the sequential antenna muting algorithm in a similar manner, as $\sum_{i=1}^{M/2} \sum_{j=1}^i P(j) F_i = \sum_{i=1}^{M/2} (2i/M) F_i$ where $F_i \in \{F_{1}, F_{2},.., F_{M/2} \}$ be the total number of FPOs required when the active antenna number is $M_i \in \{2,4,...,M\}$ respectively, and $P(i) = 2/M, \forall i$ for a uniform distribution of antenna element selection. The greedy TAM algorithm has a much higher computational complexity due to the two loops in Algorithm 1. Following the previous notations and assumptions, we can calculate the total number of FPOs as $\sum_{i=1}^{M/2} \sum_{j=1}^i P(j) (M/2+1-i) F_i = \sum_{i=1}^{M/2} 2i/M  (M/2+1-i) F_i$. The extra $(M/2+1-i)$ factor is due to the inner \textit{for} loop in Algorithm 1 where the best element from the remaining set of antennas is selected for a given candidate antenna number $M_i$.

On the other hand, in NAM the NN has a fixed computational complexity which is only dependent on $M$ and output classes $N$. The computational complexity of a  convolutional layer and a dense layer is as follows. In order to compute the output of a convolutional layer with input size $[x_1, x_2, n_i]$, using a kernel of size $[a, b]$ with $n_k$ channels, it requires $[(a b n_i n_k) (x_1 - a + 1) (x_2 - b + 1)]$ multiplications and $[(a b n_i - 1 + 1) (x_1 - a + 1) (x_2 - b + 1) n_k]$ additions. Thus, a convolutional layer requires a total of $[2 (a b n_i n_k) \times (x_1 - a + 1) (x_2 - b + 1)]$ number of FPOs. A dense layer with $A$ neurons and $B$ input dimension requires $A B$ multiplications. It requires $(B-1)$ additions for adding the multiplied values and then 1 more addition to add the bias in each neuron, thus resulting in total $A B$ additions per layer. Therefore, the total number of FPOs for a dense layer adds up to $2AB$. Apart from the computations to produce the NN output, the proposed NAM algorithm has an extra one-time computational complexity to prepare the input; specifically to produce the channel covariance matrix and the Eigen-beamforming vectors of $K$ users considering the full antenna configuration. Channel covariance matrix calculation and Eigen-beamforming vector calculation each have a $\mathcal{O}(M^3)$ complexity. Furthermore, each convolutional layer has $(abn_i + 1) n_k$ number of learnable parameters, and each dense layer has $A(B+1)$ learnable parameters which give an idea about the training complexity of the model.

\begin{table}[ht]
\caption{Computational complexity of different algorithms}
\begin{center}
\begin{tabular}{|l|l|}
\hline
\textbf{Algorithm} & \textbf{Computational complexity (Number of FPOs} \\ 
& \textbf{per channel instance)} \\ \hline
Fixed column & $ \sum_{i=0}^{N-1} i/N \times F_i$, \\ 
antenna muting & where $F_i \in \{F_0, F_{1},.., F_{7}\}$ when the active \\
& antenna number is $M_i \in \{8,16,24,...,64\}$ \\ \hline

Sequential antenna & $ \sum_{i=1}^{M/2} 2i/M \times F_i$, \\ 
muting & where $F_i \in \{F_{1}, F_{2},.., F_{M/2} \}$ when the active \\
& antenna number is $M_i \in \{2,4,...,M\}$ \\ \hline

Greedy TAM & $\sum_{i=1}^{M/2} 2i/M \times (M/2+1-i) \times F_i$, \\ 
& where $F_i$ is the same as in seq. antenna muting \\ \hline

NAM: convolutional & $2 (a b n_i n_k) \times (x_1 - a + 1) \times (x_2 - b + 1)$ \\
layer & (input size $[x_1, x_2, n_i]$, kernal size $[a, b, n_k]$) \\ \hline
NAM: dense layer & $2 \times A \times B$ ($A$ neurons and $B$ input dimension) \\ \hline

\end{tabular}
\end{center}
\label{tab:complexity}
\end{table}

\begin{figure}[ht]
\centerline{\includegraphics [width=0.52\textwidth]{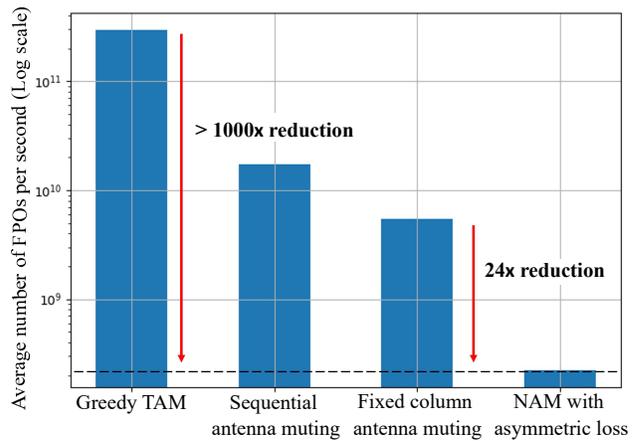}}
\caption{Average number of FPOs (per second) required for different TAM strategies (y-axis is in log scale).}
\label{fig:avg_flop_count}
\end{figure}

Fig. \ref{fig:avg_flop_count} shows the number of FPOs (per second) in average spent to obtain the antenna muting result in NAM and different heuristic TAM algorithms when BS is having $M=64$ antennas and assuming a uniform distribution of antenna classes. The greedy TAM algorithm has the highest complexity which is more than $1000\times$ computationally expensive than the NN. Furthermore, the proposed NAM approach is around $24\times $ faster than its baseline fixed column TAM strategy, making NAM a potential solution for practical implementation.

\section{Conclusion}

In this paper, we have exploited the potential of TAM in achieving energy savings in massive MIMO networks. We have proposed NAM, an ML-based solution to learn TAM in a data-driven manner to minimize the BS power consumption by activating a minimum number of antennas while guaranteeing the user QoS requirements. Furthermore, we have presented a greedy-based TAM algorithm and several other heuristics and compared them along with the performance and computational complexity of the proposed NAM approach. In NAM, a NN is trained in a supervised manner to learn the TAM given from a baseline algorithm. We have introduced a new custom loss function that improves the reliability of the NN output in achieving the QoS requirements. The proposed NAM solution achieves $ \sim73\%$ energy saving compared to the full antenna configuration in the BS with $\sim95\%$ reliability in achieving the user throughput requirements while being around $24\times$ less computationally intensive than its counterpart fixed column TAM strategy.

While we have not done extensive investigations on optimizing the NN model architecture, hyper-parameter tuning, and tuning the parameters $\lambda$ and $\alpha$ in the custom loss function, the potential of the NN approach is still evident even with the currently used configurations. Further optimizations could result in achieving improved reliability and further reductions in computational complexity. The next steps in this work would be to investigate its scalability when the number of antennas $M$ is increasing and to implement an unsupervised learning approach to enable the NN to learn better antenna selections than the sub-optimal baseline labels that are currently used. Furthermore, given that about $50\%$ of the complexity in NAM is due to the input preparation which involves channel covariance matrix calculation and Eigen-beamforming computation, further complexity reductions could be achieved if a better input could be used for the NN or if the input preparation could be simplified by using some approximations.

 \section{Acknowledgement}
This work is supported by Hexa-X-II Project, the second phase of the European 6G Flagship Initiative.

\bibliography{ref}

\begin{thebibliography}{10}
\providecommand{\url}[1]{#1}
\csname url@samestyle\endcsname
\providecommand{\newblock}{\relax}
\providecommand{\bibinfo}[2]{#2}
\providecommand{\BIBentrySTDinterwordspacing}{\spaceskip=0pt\relax}
\providecommand{\BIBentryALTinterwordstretchfactor}{4}
\providecommand{\BIBentryALTinterwordspacing}{\spaceskip=\fontdimen2\font plus
\BIBentryALTinterwordstretchfactor\fontdimen3\font minus
  \fontdimen4\font\relax}
\providecommand{\BIBforeignlanguage}[2]{{%
\expandafter\ifx\csname l@#1\endcsname\relax
\typeout{** WARNING: IEEEtran.bst: No hyphenation pattern has been}%
\typeout{** loaded for the language `#1'. Using the pattern for}%
\typeout{** the default language instead.}%
\else
\language=\csname l@#1\endcsname
\fi
#2}}
\providecommand{\BIBdecl}{\relax}
\BIBdecl

\bibitem{Energy2011}
A.~Fehske, G.~Fettweis, J.~Malmodin, and G.~Biczok, ``The global footprint of
  mobile communications: The ecological and economic perspective,'' \emph{IEEE
  Communications Magazine}, vol.~49, no.~8, pp. 55--62, 2011.

\bibitem{6736761_mimo}
E.~G. Larsson, O.~Edfors, F.~Tufvesson, and T.~L. Marzetta, ``{Massive {MIMO}
  for next generation wireless systems},'' \emph{IEEE Communications Magazine},
  vol.~52, no.~2, pp. 186--195, 2014.

\bibitem{Jacob2011}
J.~Hoydis, S.~ten Brink, and M.~Debbah, ``Massive {MIMO}: How many antennas do
  we need?'' in \emph{2011 49th Annual Allerton Conference on Communication,
  Control, and Computing (Allerton)}, 2011, pp. 545--550.

\bibitem{8644331_nokia}
H.~Halbauer, A.~Weber, D.~Wiegner, and T.~Wild, ``{Energy efficient massive
  MIMO array configurations},'' in \emph{2018 IEEE Globecom Workshops (GC
  Wkshps)}, 2018, pp. 1--6.

\bibitem{6824974_xiang}
X.~Gao, O.~Edfors, J.~Liu, and F.~Tufvesson, ``{Antenna selection in measured
  massive MIMO channels using convex optimization},'' in \emph{2013 IEEE
  Globecom Workshops (GC Wkshps)}, 2013, pp. 129--134.

\bibitem{7569725_gao}
Y.~Gao and T.~Kaiser, ``{Antenna selection in massive {MIMO} systems:
  Full-array selection or subarray selection?}'' in \emph{2016 IEEE Sensor
  Array and Multichannel Signal Processing Workshop (SAM)}, 2016, pp. 1--5.

\bibitem{7127500_Benmimoune}
M.~Benmimoune, E.~Driouch, W.~Ajib, and D.~Massicotte, ``{Joint transmit
  antenna selection and user scheduling for massive MIMO systems},'' in
  \emph{2015 IEEE Wireless Communications and Networking Conference (WCNC)},
  2015, pp. 381--386.

\bibitem{9172108_Akhtar}
J.~Akhtar, K.~Rajawat, V.~Gupta, and A.~K. Chaturvedi, ``Joint user and antenna
  selection in massive-{MIMO} systems with {QoS}-constraints,'' \emph{IEEE
  Systems Journal}, vol.~15, no.~1, pp. 497--508, 2021.

\bibitem{7959971_dong}
Y.~Dong, Y.~Tang, and K.~Zhang, ``{Improved joint antenna selection and user
  scheduling for massive MIMO systems},'' in \emph{2017 IEEE/ACIS 16th
  International Conference on Computer and Information Science (ICIS)}, 2017,
  pp. 69--74.

\bibitem{6290313_jiang}
C.~Jiang and L.~J. Cimini, ``Antenna selection for energy-efficient {MIMO}
  transmission,'' \emph{IEEE Wireless Communications Letters}, vol.~1, no.~6,
  pp. 577--580, 2012.

\bibitem{Arash_2017}
\BIBentryALTinterwordspacing
M.~Arash, E.~Yazdian, M.~S. Fazel, G.~Brante, and M.~A. Imran, ``Employing
  antenna selection to improve energy efficiency in massive {MIMO} systems,''
  \emph{Transactions on Emerging Telecommunications Technologies}, vol.~28,
  no.~12, p. e3212, aug 2017. [Online]. Available:
  \url{https://doi.org/10.1002}
\BIBentrySTDinterwordspacing

\bibitem{7523998_jingon}
J.~Joung, ``Machine learning-based antenna selection in wireless
  communications,'' \emph{IEEE Communications Letters}, vol.~20, no.~11, pp.
  2241--2244, 2016.

\bibitem{8924932_elbir}
A.~M. Elbir and K.~V. Mishra, ``Joint antenna selection and hybrid beamformer
  design using unquantized and quantized deep learning networks,'' \emph{IEEE
  Transactions on Wireless Communications}, vol.~19, no.~3, pp. 1677--1688,
  2020.

\bibitem{8446042_ibrahim}
M.~S. Ibrahim, A.~S. Zamzam, X.~Fu, and N.~D. Sidiropoulos, ``Learning-based
  antenna selection for multicasting,'' in \emph{2018 IEEE 19th International
  Workshop on Signal Processing Advances in Wireless Communications (SPAWC)},
  2018, pp. 1--5.

\bibitem{9337188_Thang}
T.~X. Vu, S.~Chatzinotas, V.-D. Nguyen, D.~T. Hoang, D.~N. Nguyen, M.~D. Renzo,
  and B.~Ottersten, ``Machine learning-enabled joint antenna selection and
  precoding design: {From} offline complexity to online performance,''
  \emph{IEEE Transactions on Wireless Communications}, vol.~20, no.~6, pp.
  3710--3722, 2021.

\bibitem{9448468_yu}
W.~Yu, T.~Wang, and S.~Wang, ``Multi-label learning based antenna selection in
  massive {MIMO} systems,'' \emph{IEEE Transactions on Vehicular Technology},
  vol.~70, no.~7, pp. 7255--7260, 2021.

\bibitem{1312557_mu_mimo_mmse}
A.~Tenenbaum and R.~Adve, ``Joint multiuser transmit-receive optimization using
  linear processing,'' in \emph{2004 IEEE International Conference on
  Communications (IEEE Cat. No.04CH37577)}, vol.~1, 2004, pp. 588--592.

\bibitem{4355332_mmse_transceiver}
S.~Shi, M.~Schubert, and H.~Boche, ``{Downlink MMSE transceiver optimization
  for multiuser MIMO systems: Duality and sum-MSE minimization},'' \emph{IEEE
  Transactions on Signal Processing}, vol.~55, no.~11, pp. 5436--5446, 2007.

\bibitem{Väisänen2018_eig_su_mimo}
\BIBentryALTinterwordspacing
N.~Väisänen, ``\BIBforeignlanguage{English}{Beamforming techniques for
  optimizing channel capacity in wireless communications},'' Master's thesis,
  Aalto University. School of Science, 2018. [Online]. Available:
  \url{http://urn.fi/URN:NBN:fi:aalto-201802231655}
\BIBentrySTDinterwordspacing

\bibitem{1603708_sus_scheduling}
T.~Yoo and A.~Goldsmith, ``On the optimality of multiantenna broadcast
  scheduling using zero-forcing beamforming,'' \emph{IEEE Journal on Selected
  Areas in Communications}, vol.~24, no.~3, pp. 528--541, 2006.

\bibitem{1207369_multiantenna_gbc_zf}
G.~Caire and S.~Shamai, ``{On the achievable throughput of a multiantenna
  Gaussian broadcast channel},'' \emph{IEEE Transactions on Information
  Theory}, vol.~49, no.~7, pp. 1691--1706, 2003.

\bibitem{1261332_zeroforcing}
Q.~Spencer, A.~Swindlehurst, and M.~Haardt, ``{Zero-forcing methods for
  downlink spatial multiplexing in multiuser MIMO channels},'' \emph{IEEE
  Transactions on Signal Processing}, vol.~52, no.~2, pp. 461--471, 2004.

\bibitem{1177729959_sherman_morrison}
\BIBentryALTinterwordspacing
``Adjustment of an inverse matrix corresponding to changes in the elements of a
  given column or a given row of the original matrix (abstract),'' \emph{The
  Annals of Mathematical Statistics}, vol.~20, no.~4, pp. 620 -- 624, 1949.
  [Online]. Available: \url{https://doi.org/10.1214/aoms/1177729959}
\BIBentrySTDinterwordspacing

\bibitem{1177729893_sherman_morrison_1}
\BIBentryALTinterwordspacing
J.~Sherman and W.~J. Morrison, ``Adjustment of an inverse matrix corresponding
  to a change in one element of a given matrix,'' \emph{The Annals of
  Mathematical Statistics}, vol.~21, no.~1, pp. 124 -- 127, 1950. [Online].
  Available: \url{https://doi.org/10.1214/aoms/1177729893}
\BIBentrySTDinterwordspacing

\bibitem{10.2307/4145217_cauchy_interlace}
\BIBentryALTinterwordspacing
S.-G. Hwang, ``Cauchy's interlace theorem for eigenvalues of hermitian
  matrices,'' \emph{The American Mathematical Monthly}, vol. 111, no.~2, pp.
  157--159, 2004. [Online]. Available:
  \url{http://www.jstor.org/stable/4145217}
\BIBentrySTDinterwordspacing

\bibitem{3gpp_36.873_channel_model}
V.~3GPP TR~36.873, ``{Study on 3D channel model for LTE (Release 12)},''
  \emph{3rd Generation Partnership Project (3GPP)}, Dec. 2017.

\bibitem{tensorflow}
\BIBentryALTinterwordspacing
M.~Abadi \emph{et~al.}, ``{TensorFlow: Large-scale machine learning on
  heterogeneous systems},'' 2015, software available from tensorflow.org.
  [Online]. Available: \url{https://www.tensorflow.org/}
\BIBentrySTDinterwordspacing

\bibitem{eig_complexity}
\BIBentryALTinterwordspacing
V.~Y. Pan and Z.~Q. Chen, ``The complexity of the matrix eigenproblem,'' in
  \emph{Proceedings of the Thirty-First Annual ACM Symposium on Theory of
  Computing}, ser. STOC '99.\hskip 1em plus 0.5em minus 0.4em\relax New York,
  NY, USA: Association for Computing Machinery, 1999, p. 507–516. [Online].
  Available: \url{https://doi.org/10.1145/301250.301389}
\BIBentrySTDinterwordspacing

\end{thebibliography}
\bibliographystyle{IEEEtran}

\vspace{12pt}

\end{document}